\documentclass[english]{article}
\usepackage{geometry}
\usepackage{float}
\usepackage{mathrsfs}
\usepackage{amsmath}
\usepackage{amssymb}
\usepackage{graphicx}
\usepackage{amsfonts}
\usepackage{amsthm}
\usepackage{epsfig}
\usepackage{comment}
\usepackage{babel}\date{}
\usepackage{subfig}\usepackage[all]{xy}
\usepackage{color}
\makeatletter

\@ifundefined{definecolor}{\@ifundefined{definecolor}
 {\@ifundefined{definecolor}
 {\usepackage{color}}{}
}{}
}{}

\newtheorem{theorem}{Theorem}[section]
\newtheorem{definition}{Definition}[section]
\newtheorem{lem}{Lemma}[section]
\newtheorem{rem}{Remark}[section]
\newtheorem{prop}{Proposition}[section]
\newtheorem{cor}{Corollary}[section]
\newcounter{hypA}
\newenvironment{hypA}{\refstepcounter{hypA}\begin{itemize}
  \item[({\bf A\arabic{hypA}})]}{\end{itemize}}
\newcounter{hypB}
\newenvironment{hypB}{\refstepcounter{hypB}\begin{itemize}
  \item[({\bf B\arabic{hypB}})]}{\end{itemize}}

\newcommand{\wtilde}{\widetilde}
\newcommand{\bbR}{\mathbb R}
\newcommand{\ess}{\mathrm{ESS}}
\newcommand{\ext}{\mathrm{Ext}}
\newcommand{\E}{\mathbb{E}}
\newcommand{\PP}{\mathbb{P}}
\newcommand{\LL}{\mathscr{L}}
\newcommand{\A}{\mathcal{A}}
\newcommand{\FF}{\mathscr{F}}
\renewcommand{\phi}{\varphi}

\newcommand{\prob}{\rightarrow_{\mathbb{P}}}
\newcommand{\weak}{\Rightarrow}

\newcommand{\Var}{\mathrm{Var}}

\newcommand{\Normal}{\mathcal{N}}
\newcommand{\bra}[1]{\langle #1 \rangle}
\newcommand{\domain}{\mbox{Dom}}
\newcommand{\mat}{\mathsf{M}}

\newcommand{\Exp}{\mathbb{E}}

\makeatother

\begin{document}

\begin{center}

{\Large \textbf{On the Convergence of Adaptive Sequential Monte Carlo Methods}}

\bigskip

BY ALEXANDROS BESKOS$^{1}$, AJAY JASRA$^{1}$, NIKOLAS KANTAS$^{2}$ \& ALEXANDRE H.~THI\'ERY$^{1}$

{\footnotesize $^{1}$Department of Statistics \& Applied Probability,
National University of Singapore, Singapore, 117546, SG.}\\
{\footnotesize E-Mail:\,}\texttt{\emph{\footnotesize staba@nus.edu.sg, staja@nus.edu.sg, a.h.thiery@nus.edu.sg}}\\
{\footnotesize $^{2}$Department of Mathematics,
Imperial College London, London, SW7 2AZ, UK.}\\
{\footnotesize E-Mail:\,}\texttt{\emph{\footnotesize n.kantas@ic.ac.uk}}
\end{center}

\begin{abstract}
In several implementations of Sequential Monte Carlo (SMC) methods it is natural, and important
in terms of algorithmic efficiency, to exploit the information of the history of the samples 
to optimally tune their subsequent propagations.
In this article we provide a carefully formulated asymptotic theory for a class of such \emph{adaptive}
SMC methods.
The theoretical framework developed here will cover, under assumptions,
several commonly used SMC algorithms \cite{chopin,jasra,schafer}. There are only limited results about the theoretical underpinning of such adaptive methods: we will bridge this gap by providing a weak law of large numbers (WLLN) and a central limit theorem (CLT) for some of these algorithms.
The latter seems to be the first result of its kind in the literature and provides a formal justification of algorithms used in many real data contexts \cite{jasra,schafer}.
We establish that for a general class of  adaptive SMC algorithms \cite{chopin}
the asymptotic variance of the estimators from the adaptive SMC method is \emph{identical} to a so-called `perfect' SMC algorithm which uses ideal proposal kernels. 
%This is due to the effect of the Monte-Carlo variability of the adapted terms being small (order $\mathcal{O}(N^{-1})$ where $N$ is the number of samples) compared to the effect of the standard non-adapted terms  (which is $\mathcal{O}(N^{-1/2})$) in an error decomposition.
Our results are supported by application on a complex high-dimensional posterior distribution associated with the Navier-Stokes model,  where adapting high-dimensional parameters of the proposal kernels is critical for the efficiency of the algorithm. 
\end{abstract}

%
%\tableofcontents
%
\section{Introduction}\label{sec:intro}
SMC methods are amongst the most widely used computational techniques in statistics, engineering, physics, finance and many other disciplines;
see \cite{doucet} for a recent overview. 
They are designed to approximate a sequence $\{ \eta_n \}_{n \geq 0}$ of probability distributions of increasing dimension or complexity. The method uses $N \geq 1$
weighted samples, or particles, generated in parallel and propagated via Markov kernels and resampling methods.
The method has accuracy which increases as the number of particles grows and is typically asymptotically exact.
Standard SMC methodology is by now very well understood with regards to its convergence properties and several consistency results have been proved \cite{chopin1,delmoral}.
SMC methods have also recently been proved to be stable in certain high-dimensional contexts \cite{beskos}.

In this article, we are concerned with \emph{adaptive} SMC methods; in an effort to improve algorithmic efficiency, the weights and/or Markov proposal kernels can depend upon the history of the simulated process. Such procedures appear in a wealth of articles including \cite{chopin,delmoralabc,jasra,schafer} and have important applications in, for example, econometrics, population genetics and data assimilation. The underlying idea of these algorithms  is that, using the particle approximation $\eta^N_n$ of the distribution $\eta_n$, one can exploit the induced information to build effective proposals or even to \emph{determine} the next probability distribution in the sequence; this is often achieved by using the expectation $\eta^N_n(\xi_{n+1})$ of a summary statistic $\xi_{n+1}$ with respect to the current SMC approximation $\eta^N_n$. In other cases, one can use the particles to determine  the next distribution in an artificial sequence of densities; we expand upon this point below.
% in a tempering procedure according  to a pre-selected Effective Sample Size (ESS) (it is described later on, why this might be a sensible choice).
Such approaches are expected to lead to algorithms that are more efficient than their  `non-adaptive' counter-parts. Critically, such ideas also deliver more automated algorithms by reducing the number of user-specified tuning  parameters.

Whilst the literature on adaptive MCMC methods is by now well-developed e.g.~\cite{andrieu} and sufficient conditions for an adaptive MCMC algorithm to be ergodic are well-understood, the analysis of adaptive SMC algorithms is still in its infancy. 
To the best of our knowledge, a theoretical study of the consistency and fluctuation properties
of adaptive SMC algorithms is lacking 
in the current literature. This article aims at filling this critical gap in the theory of SMC methods. Some preliminary results can be found, under exceptionally strong conditions, in \cite{crisan,jasra}. Proof sketches are given in \cite{delmoralabc} with some more realistic but limited analysis in \cite{giraud}.
We are not aware of any other asymptotic analysis of these particular class of algorithms in the literature.
Contrary to adaptive MCMC algorithms, we show in this article that it is reasonable to expect most adaptive SMC methods to be asymptotically correct.

%% is the following to be kept in the introduction?? %%%
%%%%%%%%%%%%%%%%%%%%%%%%%%%%%%
\begin{comment}
The underlying idea of our approach for proving asymptotic results is to consider a `perfect' algorithm, as in \cite{giraud}, for which the proposals and/or weights use perfect information of the exact probability distribution at the previous time-point. This ideal algorithm could  also be thought as the one where the adaptive decisions are made with an infinite number of particles.
We prove
%, under assumptions,  
a WLLN and a multivariate CLT for the SMC method, also referred to as `practical algorithm' in the sequel, which approximates the perfect algorithm. 
\end{comment}
%%%%%%%%%%%%%%%%%%%%%%%%%%%%%%

\subsection{Results and Structure}
This article explores two distinct directions.
In the first part, an asymptotic analysis of a class of SMC methods with adaptive Markov kernels and weights is carried out. The second part of the article looks at the case where an additional layer of randomness is taken into account through an adaptive tempering procedure.
A weak law of large numbers (WLLN) and a central limit theorem (CLT) relevant to each situation are proved. In both cases we consider a sequence of target distributions $\{ \eta_{n} \}_{n\ge 0}$ defined on a corresponding sequence of measurable spaces  $(E_n,\mathscr{E}_n)_{n\geq 0}$. 
We write $\eta_{n}^{N}=(1/N)\sum_{i=1}^{N}\delta_{x_n^{i}}$ for the $N$-particle SMC approximation of $\eta_n$, with $\delta_{x_n}$ the Dirac measure at $x_n \in E_n$ and $\{x_n^{i}\}_{i=1}^{N} \in E_n^N$ the collection of particles at time $n\ge 0$.

In the first part of the paper, for each $n \geq 1$ we consider parametric families, indexed by a parameter $\xi\in \bbR^d$, of Markov kernels $M_{n, \xi}: E_{n-1} \times \mathscr{E}_n \to \bbR_+$  and potential functions $G_{n-1, \xi}: E_{n-1} \to \bbR_{+}$. To construct the particle approximation $\eta^N_{n}$, the \emph{practical} SMC algorithm exploits summary statistics $\xi_n: E_{n-1} \to \bbR^d$ by reweighing and propagating the particle approximation $\eta^N_{n-1}$ through the potential $G_{n,\eta^N_{n-1}({\xi_n})}$ and the Markov kernel $M_{n,\eta_{n-1}^N{(\xi_n)}}$. This is a substitute for the \emph{perfect} algorithm (as also used by \cite{giraud} and which cannot be implemented) which employs the Markov kernel $M_{n,\eta_{n-1}{(\xi_n)}}$ and weight function    
$G_{n,\eta_{n-1}({\xi_n})}$. We prove a WLLN and a CLT for both the approximation of the probability distribution $\eta_n$ and its normalising constant.
This set-up is relevant, for example, in the context of sequential Bayesian parameter inference \cite{chopin,kantas} when $\{ \eta_{n} \}_{n \geq 0}$ is a sequence of posterior distributions that corresponds to increasing amount of data. The Markov kernel $M_{n,\eta_{n-1}^{N}{(\xi_n)}}$ is user-specified 
and its role is to efficiently move the particles within the state space. In many situations the Markov kernel $M_{n,\eta_{n-1}^{N}{(\xi_n)}}$ is constructed so that it leaves the distribution $\eta_n$ invariant; a random walk Metropolis kernel that uses the estimated covariance structure of $\eta_{n-1}^{N}$ for scaling its jump proposals is a popular choice.
The case when there is also a tuned parameter in the weight function $G_{n,\eta_{n-1}^{N}({\xi_n})}$ 
is relevant to particle filters \cite{doucet}, as described in Section \ref{sec:ex_filt}.

The second part of this article investigates an adaptive tempering procedure. 
Standard MCMC methods can be inefficient for directly exploring complex probability distributions  involving high-dimensional state spaces, multi-modality, greatly varying scales, or combination thereof. It is a standard approach to introduce a bridging sequence of distributions $\{\eta_n\}_{n=0}^{n=n_*}$ between a distribution $\eta_0$ that is easy to sample from and the distribution of interest $\eta_{n_*} \equiv \pi$. In accordance with the simulated annealing literature, the probability distribution of interest is written as $\pi(dx) = Z^{-1} \, e^{-\beta_* \, V(x)} \, m(dx)$ for a potential $V$, temperature parameter $\beta_*\in \bbR$, dominating measure $m(dx)$ and normalisation constant $Z$; the bridging sequence of distributions is constructed by introducing a ladder of temperature parameters $\beta_0 \leq \beta_1 \leq \cdots \leq \beta_{n_*} =: \beta_*$ and setting
$\eta_{n}(dx) = Z(\beta_n)^{-1} \, e^{-\beta_n \, V(x)} \, m(dx)$
for a normalisation constant $Z(\beta_n)$. The choice of the bridging sequence of distributions is an important and complex problem, see e.g.~\cite{gelman}. To avoid the task of having to pre-specify a potentially large number of temperature parameters, an adaptive SMC method can compute them `on the fly' \cite{jasra,schafer}, thus obtaining a random increasing sequence of temperature parameters 
$\big\{ \beta_n^N \big\}_{n \geq 0}$. In this article, we adopt the following strategy: assuming a particle approximation $\eta_{n-1}^N = (1/N) \sum_{i=1}^N \delta_{x^i_{n-1}}$ with temperature parameter $\beta_{n-1}^{N}$, the particles are assigned weights proportional to $e^{-(\beta^N_n - \beta^N_{n-1}) \, V(x^i_{n-1})}$ to represent the next distribution in the sequence; the choice of $\beta^N_{n}$ is determined from the particle collection $\{x_{n-1}^{i}\}_{i=1}^{N}$ by ensuring a minimum effective sample size (ESS)  (it is described later on, why this might be a sensible choice). This can efficiently be implemented using a bisection method; see e.g.~\cite{jasra}. We prove a WLLN and a CLT for both the approximation of the probability distribution $\eta_n$ and the estimates of the normalising constants $Z(\beta_n)$.

%Our work will prove a WLLN and CLT for these two algorithmic directions.
%We will describe in more details the above sketched algorithms and the related applications in the sequel. 

%\subsubsection*{Sequential Bayesian Parameter Inference.}
%To formalize a little, suppose that we are interested in a sequence of probabilities on a fixed state-space $E$, (as in e.g.~\cite{chopin}), which we denote by $(\eta_n)_{n\geq 0}$ and
%that one only adapts the proposals propagating the particles, which are Markovian and depend only on the previous measure; we denote these as $(M_{n,\eta_{n-1}(\xi_n)})_{n\geq 1}$, with $\xi_n$ a $d$-dimensional statistic.
%To complete the specification, we denote the importance weights as $(G_n)_{n\geq 0}$. 
%
%

One of the contributions of the article is the proof that the asymptotic variance in the CLT,  for some algorithms in the first part of the paper, is \emph{identical} to the one of the `perfect' SMC algorithm using the ideal kernels.
%That is, the perfect algorithm uses $N$ particles but has access to the perfect information which, in practice, has to be estimated adaptively. 
%This is because the effect of the Monte-Carlo variability of the adapted terms is small (order $\mathcal{O}(N^{-1}))$ compared to the effect of the standard non-adapted terms  (which is $\mathcal{O}(N^{-1/2})$) in an error decomposition.
One consequence of this effect is that if the asymptotic variance associated to the (relative) normalizing constant estimate increases linearly with respect to time (see e.g.~\cite{cerou}), then so does the asymptotic variance for the adaptive algorithm. We present numerical results on a complex high-dimensional posterior distribution associated with the Navier-Stokes model  (as in e.g.~\cite{kantas}), where adapting the proposal kernels over hundreds of different directions is critical for the efficiency of the algorithm.  Whilst our theoretical result (with regards to the asymptotic variance) only holds for the case where one adapts the proposal kernel, the numerical application will involve much more advanced adaptation procedures. These experiments provide some evidence that our theory could be relevant 
in more general scenarios.

This article is structured as follows. In Section \ref{sec:algo} the adaptive SMC algorithm is introduced and the associated notations are detailed. 
In Section \ref{sec:exam1} we provide some motivating examples for the use of adaptive SMC.
In Section \ref{sec:main_res} 
we study the asymptotic properties 
of a class of SMC algorithms with adaptive Markov kernels and weights.
In Section \ref{sec:annealed},
we extend our analysis to the case where an adaptive tempering scheme is taken into account. In each situation, we prove  a WLLN and a CLT.
%To the best of our knowledge, these results have not appeared before in the literature.
In Section \ref{sec:exam}, we verify that our assumptions 
hold when using the adaptive SMC algorithm in a 
real scenario. In addition, we provide numerical results associated to the Navier-Stokes model and some theoretical insights associated to
the effect of the dimension of the statistic which is adapted.
%, which illustrates the impact of our theoretical results.
The article is concluded in Section \ref{sec:summ} with a discussion of future work.
The appendix features a proof of one of the results in the main text.

%
%  ALGO + NOTATIONS
%
\section{Algorithm and Notations}\label{sec:algo}
In this section we provide the necessary notations and describe the SMC algorithm with adaptive Markov kernels and weights. The description of the adaptive tempering procedure is postponed to Section \ref{sec:annealed}.

\subsection{Notations and definitions}

Let $(E_n,\mathscr{E}_n)_{n\geq 0}$ be a sequence of measurable spaces endowed with a countably generated $\sigma$-field $\mathscr{E}_n$. The set $\mathcal{B}_b(E_n)$ denotes the class of bounded $\mathscr{E}_n/\mathbb{B}(\bbR)$-measurable functions on $E_n$ where $\mathbb{B}(\mathbb{R})$ Borel $\sigma$-algebra on $\mathbb{R}$. The supremum norm is written as $\|f\|_{\infty} = \sup_{x\in E_n}|f(x)|$ and $\mathcal{P}(E_n)$ is the set of probability measures on $(E_n,\mathscr{E}_n)$. We will consider non-negative operators $K : E_{n-1} \times \mathscr{E}_n \rightarrow \bbR_+$ such that for each $x \in E_{n-1}$ the mapping $A \mapsto K(x, A)$ is a finite non-negative measure on $\mathscr{E}_n$ and for each $A \in \mathscr{E}_n$ the function $x \mapsto K(x, A)$ is $\mathscr{E}_{n-1} / \mathbb{B}(\mathbb{R})$-measurable; the kernel $K$ is Markovian if $K(x, dy)$ is a probability measure for every $x \in E_{n-1}$.
%A kernel induces two integral operators, the first acting on the space of $\sigma$-finite measures on $(E_{n-1}, \mathscr{E}_{n-1})$ and the other on $\mathcal{B}_b(E_n)$. More specifically, 
For a finite measure $\mu$ on $(E_{n-1},\mathscr{E}_{n-1})$ and Borel test function $f \in \mathcal{B}_b(E_n)$ we define
\begin{equation*}
    \mu K  : A \mapsto \int K(x, A) \, \mu(dx)\ ;\quad 
    K f :  x \mapsto \int f(y) \, K(x, dy)\ .
\end{equation*}
%For $n\geq 1$, the $\mathscr{E}_{n-1} / \mathbb{B}(\mathbb{R})$-measurable function $G_{n-1}:E_{n-1} \to \bbR$ and Markov kernel $M_n:E_{n-1} \times \mathscr{E}_n \rightarrow \bbR_+$, one can define the non-negative operator $Q_n:  E_{n-1} \times \mathscr{E}_n \rightarrow \bbR_+$ by $Q_n(x,A) := G_{n-1}(x) M_{n}(x,A)$.
We will use the following notion of continuity at several places in this article.
\begin{definition} \label{def.unif.cont}
Let $\mathcal{X}$, $\mathcal{Y}$ and $\mathcal{Z}$ be three metric spaces. A function $f: \mathcal{X} \times \mathcal{Y} \to \mathcal{Z}$ is continuous at $y_0\in \mathcal{Y}$ uniformly on $\mathcal{X}$ if 
\begin{equation} \label{eq.uniform.cont}
\limsup_{\delta \to 0^+} \;  \Big\{ d_{\mathcal{Z}}
\big(f(x,y), f(x,y_0) \big) \; : \; x \in \mathcal{X}, \; d_{\mathcal{Y}}(y,y_0) < \delta \Big\} = 0\ .
\end{equation}
\end{definition}
%
%
%We use the notation $o_{\mathcal{\mathbb{P}}}(1)$ to denote terms that converge,  in  probability, to zero. 
%For a function $\phi:E\rightarrow\mathbb{R}$ we use the notation $\|\phi\|_{\infty}=\sup_{x\in E}|\phi(x)|$. A vector or matrix is bounded (resp.~continuous) when each element of the vector or matrix is bounded (resp.~continuous).
We write $\rightarrow_{\mathbb{P}}$  and $\weak$ to denote convergence in probability and in distributions. The Kroenecker product $u \otimes v$ of two vectors $u,v \in \bbR^d$ designates the matrix $u \cdot v^{\top} \in
\mathbb{R}^{d\times d}$; the covariance of a function $\phi\in\mathcal{B}_b(E)^r$ with respect to a probability measure $\mu\in\mathcal{P}(E)$ is denoted by  
$\Sigma_{\mu}(\phi) 
=
\int_E [\phi(x)-\mu(\phi)] \otimes [\phi(x)-\mu(\phi)] \, \mu(dx)$.
%The notation $\mathcal{N}_d(\mu,\Sigma)$ denotes a $d$-dimensional Gaussian distribution with mean $\mu$ and covariance matrix $\Sigma$. 
%The set of real $p \times q$ matrices is denoted by $\mat_{p,q}(\bbR)$. 

\subsection{SMC Algorithm}\label{sec:algo1}

%We begin with the following definitions. Let $(E_n,\mathscr{E}_n)_{n\geq 0}$ be a sequence of measurable spaces. For $\phi\in\mathcal{B}_b(E_n)$ (the Banach-space of scalar  bounded and measurable functions), we define:
For each index $n \geq 1$, we consider Markov operators $M_{n,\xi}: E_{n-1} \times \mathscr{E}_n \rightarrow \bbR_+$ and weight functions $G_{n-1, \xi}: E_{n-1} \to \bbR_+$ parametrized by $\xi \in \bbR^d$. The adaptive SMC algorithm to be described exploits summary statistics $\xi_n: E_{n-1} \to \bbR^d$ and aims at approximating the sequence of probability distributions $\{\eta_n\}_{n \geq 0}$, on the measurable spaces $(E_n,\mathscr{E}_n)_{n\geq 0}$,
defined via their operation on a test function $\phi_n \in \mathcal{B}_b(E_n)$ as
\begin{equation}
\label{eq:eta_def}
\eta_n(\phi) := \gamma_n(\phi_n) / \gamma_n(1)
\end{equation}
where $\gamma_n$ is the unnormalised measure on $(E_n,\mathscr{E}_n)$ given by 
\begin{equation}
\label{eq:gamma_def}
\gamma_n(\phi) := \mathbb{E}\,\Big[\prod_{p=0}^{n-1} G_p(X_p)\cdot\phi(X_n)\Big]\ .\end{equation}
The above expectation is under the law of a non-homogeneous Markov chain $\big\{ X_n \big\}_{n \geq 0}$ with initial distribution $X_0 \sim \eta_0 \equiv \gamma_0$ and transition $\PP\,[\,X_{n} \in A \mid X_{n-1}=x\,] = M_n(x,A)$ where we have used the notations
\begin{equation*}
M_n \equiv M_{n, \eta_{n-1}(\xi_n)}\ ; \quad 
G_{n} \equiv G_{n, \eta_{n}(\xi_{n+1})}\ .
\end{equation*}
In practice, the expectations $\eta_{n-1}(\xi_n)$ of the summary statistics are not analytically tractable and it is thus impossible to simulate from the Markov chain $\{X_n\}_{n \geq 0}$ or compute the weights $G_n$. Nevertheless, for the purpose of analysis, we introduce the following idealized algorithm, referred to as the \emph{perfect} SMC algorithm in the sequel, that propagates a set of $N \geq 1$ particles by sampling from the distribution
\begin{equation}
\label{eq:strange}
\mathbb{P}\big(\,d(x_0^{1:N},x_1^{1:N},\ldots,x_n^{1:N})\,\big)
=
\prod_{i=1}^N \eta_0(dx_0^i) \, \prod_{p=1}^n \prod_{i=1}^N\Phi_{p}(\eta_{p-1}^N)(dx_p^i)
\end{equation}
where the $N$-particle approximation of the distribution \eqref{eq:eta_def} is defined as
\begin{equation} \label{eq.approx.normalized.measure}
\eta_n^{N}= \frac{1}{N} \, \sum_{i=1}^N \delta_{x_n^i}\ .
\end{equation}
In \eqref{eq:strange}, the operator $\Phi_n: \mathcal{P}(E_{n-1}) \to \mathcal{P}(E_{n})$ is
\begin{equation*}
\Phi_n(\mu)(dy) = \frac{\mu(G_{n-1}M_{n})(dy)}{\mu(G_{n-1})} \ .
\end{equation*}
%
%where: 
%\begin{align*}
%\mu(G_{n-1})&=\int_{E_{n-1}} G_{n-1}(x)\mu(dx)\ ,\\ 
%\mu(G_{n-1}M_{n,\eta_{n-1}(\xi_n)})(dy) &= \int_{E_{n-1}} G_{n-1}(x)
%M_{n,\eta_{n-1}(\xi_n)}(x,dy)\mu(dx)\ .
%\end{align*}
Expression \eqref{eq:strange} is a mathematically concise way to describe a standard particle method that begins by sampling $N$ i.i.d.~particles from 
$\eta_0$ and, given particles $\{ x_{n-1}^{i}\}_{i=1}^{N}$,  performs multinomial resampling according to the unnormalised weights 
$G_{n-1}(x_{n-1}^i)$ before propagating the particles via the Markov kernel $M_{n}(x,dy)$.

The SMC algorithm that is actually simulated in practice, referred to as the \emph{practical} SMC algorithm in the sequel, has joint law
\begin{equation}
\label{eq:d1}
\mathbb{P}\,\big(\,d(x_0^{1:N},x_1^{1:N},\dots,x_n^{1:N})\,\big) = \prod_{i=1}^N \eta_0(dx_0^i) \,\prod_{p=1}^n \prod_{i=1}^N\Phi_{p,N}(\eta_{p-1}^N)(dx_p^i)\ .
\end{equation}
The operator $\Phi_{n,N}$ approximates the ideal one, $\Phi_{n}$, and is defined as
\begin{equation*}
\Phi_{n,N}(\mu)(dy) = \frac{\mu(G_{n-1,N} \, M_{n,N})(dy)}{\mu(G_{n-1,N})}\ .
\end{equation*}
We have used the short-hand notations 
\begin{equation*}
M_{n,N} \equiv M_{n, \eta^N_{n-1}(\xi_n)}
\ ; \qquad  
G_{n,N} \equiv G_{n, \eta^N_n(\xi_{n+1})}\ .
\end{equation*}
%Our analysis is associated to the practical algorithm that is actually simulated.
%
%We denote the co-ordinate mappings of the statistic as $\xi_{p,j}$, $1\le j\le d$. We will be using the notation $\eta_{p-1}^N(\xi_p)=(\eta_{p-1}^N(\xi_{p,1}),\dots,\eta_{p-1}^N(\xi_{p,d}))$. 

Throughout this article we assume that the potentials are strictly positive,
$G_{n,\xi}(x) > 0$ for all $x\in E_{n}$ and $\xi \in \bbR^d$
so that there is no possibility that the algorithm collapses. 
The particle approximation of  the unnormalised distribution \eqref{eq:gamma_def} is defined as
\begin{equation}
\label{eq.approx.unnormalized.measure}
\gamma_n^N(\phi_n) = \Big\{ \prod_{p=0}^{n-1} \eta_p^N(G_{p,N}) \Big\} \, \eta^N_n(\phi_n)\ .
\end{equation}
It will bel useful to introduce the non-negative operator
\begin{equation} \label{def.operator.Q}
Q_{n,N}(x,dy) = G_{n-1,N}(x) M_{n,N}(x,dy)
\end{equation}
and the idealised version
\begin{equation*}
Q_{n}(x,dy) = G_{n-1}(x) M_{n}(x,dy) 
\equiv G_{n-1,\eta_{n-1}(\xi_n)}(x) M_{n,\eta_{n-1}(\xi_n)}(x,dy)\ .
\end{equation*}
Many times we will be interested in the properties of involved operators as functions of $\xi$, thus we will also write $$Q_{n,\xi}(x,dy) := G_{n-1,\xi}(x) \, M_{n,\xi}(x,dy)$$ to emphasise the dependency on the parameter $\xi \in \bbR^d$. Unless otherwise stated, the differentiation operation $\partial_{\xi}$ at step $n$ is evaluated at the limiting parameter value $\xi=\eta_{n-1}(\xi_{n})$.
With these definitions, one can verify that the following identities hold
\begin{equation}
\label{eq:defPhi}
\eta_n(\phi_n) = \Phi_{n}(\eta_{n-1})(\phi_n) = \frac{\eta_{n-1}(Q_{n} \phi_n)}{\eta_{n-1}(G_{n-1})}\ ;\quad 
\gamma_n(\phi_n) =  \gamma_{n-1}(Q_n \phi_n)\ .
\end{equation}
Similar formulae are available for the $N$-particle approximations; if $\mathscr{F}_{n}^N$ designates the filtration generated by the particle system up-to (and including) time $n$ we have
\begin{equation}
%\label{eq:defPhi}
\E\,\big[\eta^N_n(\phi_n) \mid \mathscr{F}_{n-1}^N \big] = \Phi_{n,N}(\eta^N_{n-1})(\phi_n)\ ; \quad 
\E\,\big[\gamma^N_n(\phi_n)  \mid \mathscr{F}_{n-1}^N  \big] =  \gamma^N_{n-1}(Q_{n,N} \phi_n)\ .
\end{equation}
In the sequel, we will use the expressions $\E_{n-1}[\, \cdot\, ]$ and $\Var_{n-1}[ \,\cdot\, ]$ to denote the conditional expectation $\E\,[\, \cdot \mid \FF^N_{n-1}\,]$ and conditional variance $\Var\,[\,\cdot \mid \FF^N_{n-1}\,]$ respectively.

%\red{(be clearer here - can be confusing)} 
%When necessary, we will write the operators  $Q_{n}$ and $Q_{n,N}$ as $Q_{n,\eta_{n-1}(\xi_n)}$ and $Q_{n,\eta_{n-1}^N(\xi_n)}$ respectively in order to emphasize the dependency on the statistics $\xi_n$. 
%Similarly, we will sometimes write $Q_{n,\eta}(x,dy) := G_n(x) \, M_{n,\eta}(x,dy)$ to emphasize the dependency on $\eta \in \bbR^d$.

\begin{rem}
Our results concern multinomial resampling at each time. 
Extension of our analysis to adaptive resampling \cite{delm:12} is possible but would require many additional calculations and technicalities; this is left as a topic for future work.
\end{rem}

%Also, for $p<n$ we will adopt the notation:
%%
%$$Q_{p:n}(\phi)(x) = \int Q_{p+1}(x,dx_{p+1})\otimes\cdots\otimes Q_{n}(x_{n-1},dx_{n})\,\phi(x_{n})\ , $$
%%
%with $Q_{p:p}$ the identity operator.
%and one should understand the distinction of this abuse of notation when it is adopted.

%
%  MOTIVATING EXAMPLES
%
\section{Motivating Examples}\label{sec:exam1}

%
% sequential Bayesian parameter estimation
%
\subsection{Sequential Bayesian Parameter Inference}
\label{sec:seq_bi}
Consider Bayesian inference for the parameter $x \in E$, observations $y_i \in \mathcal{Y}$ and prior measure $\eta_0(dx)$. The posterior distribution $\eta_n$ after having observed $y_{1:n} \in \mathcal{Y}^{n+1}$ reads
\begin{equation*}
\eta_n(dx) = \big(\, \PP\,[\,y_{1:n} \mid x\,] \,/\, \PP\,[\,y_{1:n}\,]\, \big) \, \eta_0(dx)\, .
\end{equation*}
The approach in \cite{chopin} fits in the framework described in Section \ref{sec:algo1} with state spaces $E_n=E$ and potential functions $G_n(x)= \PP\,[\,y_{n+1} \mid y_{1:n},x\,]$.
For an MCMC kernel $M_n \equiv M_{n,\eta_{n-1}(\xi_n)}$ with invariant measure $\eta_n$ the posterior distribution $\eta_n$ is given by $\eta_n(\phi_n) = \gamma_n(\phi_n) / \gamma_n(1)$
where the unnormalised measure $\gamma_n$ is defined as in \eqref{eq:gamma_def}. 
A popular choice consists in choosing for $M_{n,\eta_{n-1}(\xi_n)}$ a random walk Metropolis kernel reversible with respect to $\eta_{n}$ and jump covariance structure matching the one of the distribution $\eta_{n-1}$. Under our assumptions, the analysis of Section \ref{sec:main_res} applies in this context.

Whilst such an example is quite simple it is indicative of more complex applications in the literature.
Article \cite{kantas} considers a state-space with dimension of about $10^4$ and dimension of adapted statistic of about $500$.
In such a setting, pre-specifying the covariance structure of the random walk Metropolis proposals is impractical; the adaptive SMC strategy of Section~\ref{sec:algo} provides a principled framework for automatically setting this covariance structure, see also Section \ref{sec:num_ex}.

%
%{\it The following could be clearer: \\ *** \\
%In a slightly different context, the article \cite{schafer} seeks to perform Bayesian variable selection in a regression model; in this case the state space is finite, $E_n=\{0,1\}^k$, with a potential large value of $k \geq 1$. In the SMC algorithm, they employ
%Metropolis-Hastings steps with independent proposals parameterized in some way. The efficiency of the MCMC (which is critical to the performance of their SMC algorithm) is dependent on a suitable choice of the parameter, which is again chosen adaptively as the choice of a `good' parameter is difficult before running the algorithm. 
%\\***\\}

%
% FILTERING
%
\subsection{Filtering}\label{sec:ex_filt}
This section illustrates the case of having an adaptive weight function.
Consider a state-space model with observations $Y_{1:n} \in \mathcal{Y}^n$, unobserved Markov chain $U_{0:n}\in \mathcal{U}^{n+1}$ and joint density with respect to a dominating measure $\lambda_{\mathcal{Y}}^{\otimes n} \otimes \lambda_{\mathcal{U}}^{\otimes n+1}$ given by
\begin{equation*}
\eta_0(u_0) \prod_{p=1}^n g_p(u_p,y_p) \, f_p(u_{p-1},u_p)\  .
\end{equation*}
The probability $\eta_0(u_0) \, \lambda_{\mathcal{U}}(du_0)$ is the prior distribution for the initial state of the unobserved Markov chain, $g_p(u_p, y_p) \, \lambda_{\mathcal{Y}}(dy_p)$ is the conditional observation probability at time $p$ and $f_p(u_{p-1}, u_p) \, \lambda_{\mathcal{U}}(du_p)$ describes the dynamics of the unobserved Markov process.

A standard particle filter with proposal at time $p$ corresponding to the Markov kernel $\PP[U_{p} \in du_p \mid U_{p-1}=u_{p-1}] = m_p(u_{p-1}, u_p) \, \lambda_{\mathcal{U}}(du_p)$ has importance weights of the form
\begin{equation*}
G_p(x_p) = \frac{g_p(u_p,y_p) f_p(u_{p-1},u_p)}{m_p(u_{p-1},u_p)}
\end{equation*}
where here $x_p \equiv (x_p^{(1)}, x_p^{(2)}) \equiv (u_{p-1}, u_p)$. The process $\{X_p\}_{p=1}^n$ is Markovian with transition $M_p(x_{p-1}, dx_p) = \delta_{x_{p-1}^{(2)}}(dx_p^{(1)}) \, m_p(x_{p-1}^{(2)}, x^{(2)}_p) \, \lambda_{\mathcal{U}}(dx_p^{(2)})$. The marginals of the sequence of probability distributions $\eta_n$ described in Equation \eqref{eq:eta_def} are the standard predictors.

In practice, the choice of the proposal kernel $m_n$ is critical to the efficiency of the SMC algorithm. In such settings, one may want to exploit the information contained in the distribution $\eta_{n-1}$ in order to build efficient proposal kernels. Approximating the filter mean is a standard strategy. In these cases, both the Markov kernel $M_n$ and the weight function $G_{n-1}$ depend upon the distribution $\eta_{n-1}$; this is covered by the framework adapted in Section~\ref{sec:algo}. See \cite{doucet} and the references therein for ideas associated to such approaches. 

%
%  Asymptotic Results for Adaptive SMC
%
\section{Asymptotic Results for Adaptive SMC via Summary Statistics}\label{sec:main_res}
In this section we develop an asymptotic analysis of the class of adaptive SMC algorithm described in section \ref{sec:algo}.
After first stating our assumptions in Section \ref{sec.algo1.assump}, we give a WLLN in Section \ref{sec.algo1.WLLN} and a CLT in Section \ref{sec.algo1.CLT}.

%
% Assumption
%
\subsection{Assumptions} \label{sec.algo1.assump}
Our results will make use of conditions (A\ref{hyp:1}-\ref{hyp:2}) below. 
By $\domain(\xi_n) \subset \mathbb{R}^d$ we denote a convex set that contains the range of the statistic $\xi_n:E_{n-1} \to \mathbb{R}^{d}$.
%
% hypothesis 1
%
\begin{hypA}
\label{hyp:1}
For each $n \geq 0$, function $(x,\xi) \mapsto G_{n, \xi}(x)$ is bounded and continuous at $\xi=\eta_{n}(\xi_{n+1})$ uniformly over $x\in E_n$. Statistics $\xi_{n+1}:E_n \to \bbR^d$ are bounded. For any test function $\phi_{n+1} \in \mathcal{B}_b(E_{n+1})$ the function $(x,\xi) \mapsto Q_{n+1, \xi} \phi_{n+1} (x)$ is bounded, continuous at 
$\xi=\eta_{n}(\xi_{n+1})$ uniformly over $x\in E_n$.
\end{hypA}
%
% hypothesis 2
%
\begin{hypA}
\label{hyp:2}
For each $n \geq 0$  and test function $\phi_{n+1} \in \mathcal{B}_b(E_{n+1})$,  function 
$(x,\xi) \mapsto \partial_{\xi} Q_{n+1,\xi} \phi_{n+1}(x)$
is well defined on $E_{n} \times \domain(\xi_{n+1})$, bounded and continuous at $\xi = \eta_{n}(\xi_{n+1})$ uniformly over $x\in E_n$.
\end{hypA}
Assumptions (A\ref{hyp:1}-\ref{hyp:2}) are reasonably weak in comparison to some assumptions used in the SMC literature, such as in \cite{delmoral}, but are certainly not the weakest adopted for WLLN and CLTs (see e.g.~\cite{chopin1}). 
The continuity assumptions in (A\ref{hyp:2}) are associated to the use of a first order-Taylor expansion. We have defined $\domain(\xi_p)$ as a convex set because we need to compute integrals along segments between points of $\domain(\xi_p)$.  In general, we expect that the assumptions can be relaxed for unbounded functions at the cost of increased length and complexity of the proofs.\\

%
% WLLN section
%
\subsection{Weak Law of Large Numbers}
\label{sec.algo1.WLLN}
In this section we establish a weak law of large numbers (WLLN). To do so, we state first a slightly stronger result that will be repeatedly used in the fluctuation analysis presented in Section \ref{sec.algo1.CLT}.
%
% WLLN thm
%
\begin{theorem}
\label{theo:wlln}
Assume (A\ref{hyp:1}). Let $\mathsf{V}$ be a Polish space and $\{V_N\}_{N \geq 0}$ a sequence of $\mathsf{V}$-valued random variables that converges in probability to $\mathsf{v} \in \mathsf{V}$. Let $n \ge 0$, $r \ge 1$ and $\phi_n: E_n \times \mathsf{V} \to \bbR^{r}$ be a bounded function continuous at $\mathsf{v} \in \mathsf{V}$ uniformly on $E_n$. The following limit holds in probability
\begin{equation*}
\lim_{N \to \infty} \eta_n^N\,[\,\phi_n(\cdot, V_N)\,]
=
\eta_n\,[\,\phi_n(\cdot, \mathsf{v})\,]\ .
\end{equation*}
\end{theorem}
\begin{cor}[WLLN]
\label{cor:wlln}
Assume (A\ref{hyp:1}). Let $n \ge 0$, $r \ge 1$ and $\phi_n: E_n \to \bbR^{r}$ a bounded measurable function. The following limit holds in probability,
$\lim_{N \to \infty} \eta_n^N(\phi_n)
=
\eta_n(\phi_n)$.
\end{cor}
\begin{proof}[Proof of Theorem \ref{theo:wlln}]
It suffices to concentrate on  the scalar case $r=1$.
The proof is by induction on $n$. The initial  case $n=0$ is a direct consequence of  WLLN for i.i.d.~random variables and Definition \ref{def.unif.cont}. For notational convenience, in the rest of the proof we write $\bar{\phi}_n(\cdot)$ instead of $\phi_n(\cdot, \mathsf{v})$.
We assume the result at rank $n-1$ and proceed to the induction step.
Since $V_N$ converges in probability to $\mathsf{v} \in \mathsf{V}$, Definition \ref{def.unif.cont} shows that it suffices to prove that $[\eta_n^N-\eta_n]\big(\bar{\phi}_n\big)$ converges in probability to zero.  We use the decomposition
\begin{align*}
[\eta_n^N-\eta_n](\bar{\phi}_n) 
&= 
\big(\eta_n^N(\bar{\phi}_n)-\E_{n-1}[\eta_n^N(\bar{\phi}_n)]\big)
+
\big(\E_{n-1}[\eta_n^N(\bar{\phi}_n)]-\eta_n(\bar{\phi}_n)\big)\\
&=
[\eta_n^N-\Phi_{n,N}(\eta_{n-1}^N)](\bar{\phi}_n) 
+
[\Phi_{n,N}(\eta_{n-1}^N)-\eta_n](\bar{\phi}_n)
=: A(N) + B(N)\ .
\end{align*}
To conclude the proof, we now prove that each of these terms converges to zero in probability.
\begin{itemize}
\item
Since the expected value of $A(N)$ is zero, it suffices to prove that its moment of order two also converges to zero as $N$ goes to infinity. To this end, it suffices to notice that 
\begin{equation*}
\E_{n-1}\big[A(N)^2 \big] = \tfrac{1}{N}\,\E_{n-1}\Big[\big( \bar{\phi}(x^i_n) - \E_{n-1}[\bar{\phi}(x^i_n)] \big)^2 \Big] \leq \frac{\|\bar{\phi}\|^2_{\infty}}{N}\ .
\end{equation*}
\item
To treat the quantity $B(N)$, we use the definition of $\Phi_{n,N}(\eta_{n-1}^N)$ in \eqref{eq:defPhi} and decompose it as the sum of three terms $B(N) = B_1(N) + B_2(N) + B_3(N)$ with
\begin{align*}
B_1(N) &
=\eta_{n-1}^N \big\{ [Q_{n,N} - Q_n](\bar{\phi}_n) \big\} \, / \, \eta_{n-1}^N(G_{n-1,N})\  ;\\
B_2(N) &= [\eta_{n-1}^N-\eta_{n-1}]\big( Q_n(\bar{\phi}_n) \big) \, / \, \eta^N_{n-1}(G_{n-1,N})\  ;\\
B_3(N) &= \eta_{n-1}^N[ Q_n(\bar{\phi}_n) ] \times 
\big\{
1 / \eta^N_{n-1}(G_{n-1,N}) - 1 / \eta_{n-1}(G_{n-1})
\big\}\ .
\end{align*}
We prove that $B_i(N)$ converges in probability to zero for $i=1,2,3$. The induction hypothesis shows that $\eta_{n-1}^N(\xi_n)$ converges to $\eta_{n-1}(\xi_n)$ in probability. By Assumption~\ref{hyp:1}, the bounded function $(x,\xi) \mapsto G_{n-1, \xi}(x)$ is continuous at $\xi=\eta_{n-1}(\xi_n)$ uniformly on $E_{n-1}$; the induction hypothesis applies and $\eta^N_{n-1}(G_{n-1,N})$ converges in probability to $\eta_{n-1}(G_{n-1})$. Similarly, since $Q_n(\bar{\phi}) \in \mathcal{B}_b(E_{n-1})$ is bounded by boundedness of $\bar{\phi}_n$, it follows that $\eta^N_{n-1}[ Q_n(\bar{\phi}_n) ]$ converges in probability to $\eta_{n-1}[ Q_n(\bar{\phi}_n) ]$. Slutsky's Lemma thus yields that $B_2(N)$ and $B_3(N)$ converge to zero in probability. Finally, note that by Assumption \ref{hyp:1} the bounded function $(x,\xi) \mapsto Q_{n, \xi}(x,\bar{\phi}_n)$ is continuous at $\xi=\eta_{n-1}(\xi_n)$ uniformly on $E_{n-1}$; the induction yields 
\begin{align*}
\lim_{N \to \infty}
\eta_{n-1}^N \big\{ [Q_{n,N} - Q_n](\bar{\phi}_n) \big\}
=
\lim_{N \to \infty}
&\big\{\,\eta_{n-1}^N[Q_{n,N}(\bar{\phi}_n)] -
\eta_{n-1}[Q_n(\bar{\phi}_n)]\,\big\}\\
&-
\lim_{N \to \infty}
\big\{\,\eta_{n-1}^N[Q_{n}(\bar{\phi}_n)] -
\eta_{n-1}[Q_{n}(\bar{\phi}_n)]\,\big\} = 0\ ,
\end{align*}
which is enough for concluding that  $B_1(N)$ converges to zero in probability.
\end{itemize}
\end{proof}
As a corollary, one can establish a similar consistency result for the sequence of particle approximations $\gamma^N_n(\phi_n)$, defined in Equation \eqref{eq.approx.unnormalized.measure}, of the unnormalised quantity $\gamma_n(\phi_n)$. 
\begin{cor}
\label{cor.wlln.normalisation}
Assume (A\ref{hyp:1}). Let $n \ge 0$, $r \ge 1$ and $\phi_n: E_n \to \bbR^{r}$ be  a bounded measurable function. The following limit holds in probability,
$\lim_{N \to \infty} \gamma_n^N(\phi_n)
=
\gamma_n(\phi_n)$.
\end{cor}
\begin{proof}
Since $\gamma^N_n(\phi_n) = \gamma^N_n(1) \, \eta^N_n(1)$ and $\gamma_n(\phi_n) = \gamma_n(1) \, \eta^N_n(1)$, by Corollary \ref{cor:wlln} it suffices to prove that $\gamma^N_n(1) = \eta^N_0(G_0) \times \ldots \times \eta^N_{n-1}(G_{n-1})$ converges in probability to the value $\gamma_n(1)=\eta_0(G_0) \times \ldots \times \eta_{n-1}(G_{n-1})$. By Assumption \ref{hyp:1}, the potentials $\{ G_p \}_{p \geq 0}$ are bounded so that Corollary \ref{cor:wlln} applies and the quantity $\eta^N_{p}(G_{p})$ converges in probability to $\eta_{p}(G_{p})$ for any index $p \geq 0$. The conclusion directly follows.
\end{proof}

%
%  SECTION : CLT
%
\subsection{Central Limit Theorems}
\label{sec.algo1.CLT}
In this section, for a test function $\phi_n: E_n \to \mathbb{R}^r$, we carry out a fluctuation analysis of the particle approximations $\gamma^N_n(\phi_n)$ and $\eta^N_n(\phi_n)$ around their limiting value. As expected, we prove that there is convergence at standard Monte-Carlo rate $N^{-1/2}$; in some situations, comparison with the perfect and non-adaptive algorithm is discussed in Section \ref{sec.stability}. 
%
%  CLT thm for unormalized measure
%
\begin{theorem} \label{thm.clt.unnormalised}
Assume (A\ref{hyp:1}-\ref{hyp:2}). Let $n \ge 0$, $r \ge 1$ and $\phi_n: E_n \to \bbR^{r}$ be a bounded measurable function. The sequence $\sqrt{N} \, [\gamma^N_n - \gamma_n](\phi_n)$ converges weakly to a centered Gaussian distribution  with covariance
\begin{equation} \label{eq.asymp.var.unnormalised}
\sum_{p=0}^n \gamma_p(1)^2 \, \Sigma_{\eta_p}(\LL_{p,n} \phi_n)
\end{equation}
where the linear operator $\LL_p:\mathcal{B}_b(E_p)^r \to \mathcal{B}_b(E_{p-1})^r$ is defined by
\begin{equation} \label{eq.semigroup}
\LL_p \phi_p = 
\eta_{p-1}[\partial_{\xi} Q_p \phi_p ] \, \big( \,  \xi_p - \eta_{p-1}(\xi_p) \, \big)
+ Q_p(\phi_p)
\end{equation}
with $\LL_{p,n} := \LL_{p+1} \circ \ldots \circ \LL_{n}$ and $\LL_{n,n} = \mathrm{Id}$.
\end{theorem}
\begin{proof}
For notational convenience, we concentrate on the scalar case $r=1$. The proof of the multi-dimensional case is identical, with covariance matrices replacing scalar variances.
We proceed by induction on the parameter $n \geq 0$.
The case $n=0$ follows from the usual CLT for i.i.d.\@ random variables.
To prove the induction step it suffices to show that for any $t \in \bbR$ the following identity holds
\begin{equation}
\label{eq.induction.norm.clt}
\lim_{N \to \infty} \, \E\,[\,e^{i t \sqrt{N} \, 
[ \gamma^N_n-\gamma_n ](\phi_n)}\,]
= 
e^{-\frac12 t^2 \, \gamma_n(1)^2 \, \Sigma_{\eta_n}(\phi_n)} \, 
\lim_{N \to \infty} \,
\E\,[\,e^{i t \sqrt{N} \, [ \gamma^N_{n-1}-\gamma_{n-1} ](\LL_n \phi_n)}\,]\ .
\end{equation}
\noindent
Indeed, assuming that the induction hypothesis holds at time $n-1$, we have that
\begin{equation*}
\lim_{N \to \infty} \; \E\,[\,e^{i t \sqrt{N} \, [ \gamma^N_{n-1}-\gamma_{n-1} ](\LL_n \phi_n)}\,]
=
\exp\big\{ -\tfrac{1}{2} t^2 \, \sum_{p=0}^{n-1} \gamma_p(1)^2 \, \Sigma_{\eta_p}(\LL_{p,n} \phi_n) \big\}\ 
\end{equation*}
and the proof of the induction step then follows from Levy's continuity theorem and \eqref{eq.induction.norm.clt}. 
%We now concentrate on proving that Equation \eqref{eq.induction.norm.clt} holds. 
%Recall that $\gamma_n(\phi_n) = \gamma_{n-1}(Q_n \phi_n)$.
To prove \eqref{eq.induction.norm.clt} we use the following decomposition
\begin{align*} 
%\label{eq.gamma.decomposition}
\begin{aligned} \,
[\gamma^N_n-\gamma_n ](\phi_n)
&=
\big\{ \gamma^N_n(\phi_n) - \E_{n-1}[\gamma^N_n(\phi_n)] \big\}
+
\big\{ \E_{n-1}[\gamma^N_n(\phi_n)]
-
\gamma_n(\phi_n) \big\}\\
&=: \wtilde{A}(N) + \wtilde{B}(N)\ .
\end{aligned}
\end{align*}
Since $\wtilde{B}(N) \in \mathscr{F}^N_{n-1}$ the expectation $\E[e^{i t \sqrt{N} \, [ \gamma^N_n-\gamma_n ](\phi_n)}]$ can be decomposed as
\begin{align*} %\label{eq.characteristic.decomposition}
\begin{aligned}
\E \Big[\Big( \E_{n-1}\big[ e^{it \sqrt{N} \wtilde{A}(N)}\big] 
&-
e^{-\frac12 t^2 \, \gamma_n(1)^2 \, \Sigma_{\eta_n}(\phi_n)}
\Big) \times e^{it \sqrt{N} \, \wtilde{B}(N)} \Big]\\
&\qquad+
e^{-\frac12 t^2 \, \gamma_n(1)^2 \, \Sigma_{\eta_n}(\phi_n)} \times \E \big[ e^{it \sqrt{N} \, \wtilde{B}(N)} \big]\ .
\end{aligned}
\end{align*}
As a consequence,  \eqref{eq.induction.norm.clt} follows once it is established that the limit
\begin{equation} 
\label{eq.cv.prob.fourier} 
\lim_{N \to \infty} \E_{n-1}\Big[ e^{it \sqrt{N} \wtilde{A}(N)} \Big]
= 
\exp\big\{-\tfrac12 t^2 \, \gamma_n(1)^2 \, \Sigma_{\eta_n}(\phi_n) \big\}
\end{equation}
holds in probability and that $\sqrt{N} \, \wtilde{B}(N) = \sqrt{N} \,  [\gamma^N_{n-1} - \gamma_{n-1}](\LL_n(\phi_n)) + o_{\mathbb{P}}(1)$. 
We finish the proof of Theorem \ref{thm.clt.unnormalised} by establishing these two results.
\begin{itemize}
\item 
Quantity $\wtilde{A}(N)$ also reads as $\gamma^N_n(1) \, A(N)$ with $A(N) := \big[ \eta^N_n-\Phi_{n,N}(\eta^N_{n-1}) \big](\phi_n)$. 
By Corollary \ref{cor.wlln.normalisation}, $\gamma_n^N(1)$ converges in probability to $\gamma_n(1)$; to prove that $\E_{n-1}\big[ e^{it \sqrt{N} \wtilde{A}(N)}\big]$ converges in probability to $\exp\big\{-\frac12 t^2 \, \gamma_n(1)^2 \, \Sigma_{\eta_n}(\phi_n) \big\}$ it thus suffices to show that $\E_{n-1}\big[ e^{it \sqrt{N} A(N)}\big]$ converges in probability to $\exp\big\{-\tfrac12 t^2 \, \Sigma_{\eta_n}(\phi_n) \big\}$. We will exploit the following identity
\begin{equation*}
\E_{n-1}\,\Big[\, e^{ i \, t \, \sqrt{N} \, A(N) }\,\Big]
=
\E_{n-1}\,\Big[\,e^{ i \, t \, \{ \phi_n(X_N)- \E_{n-1}[\phi_n(X_N)] \} / \sqrt{N}}\,\Big]^N
\end{equation*}
with $X_N$ is distributed according to $\sum_{i=1}^N \frac{G_{n-1,N}(x_{n-1}^i)}{\sum_{j=1}^N G_{n-1,N}(x_{n-1}^i)} \, M_{n,N}(x_{n-1}^i, dx)$. Since the test function $\phi_n$ is bounded, a Taylor expansion yields that
\begin{equation*}
\E_{n-1}\,\Big[\, e^{ i \, t \, \{ \phi_n(X_N)- \E_{n-1}[\phi_n(X_N)] \} / \sqrt{N}}\,\Big]
=
1 - \tfrac{t^2}{N}\, \Var_{n-1}[\phi_n(X_N)] + N^{-3/2} \times \mathcal{O}_{\mathbb{P}}(1)\ .
\end{equation*}
Consequently, $\E_{n-1}[e^{i t \sqrt{N} \, A(N)}] = \exp\big\{ - t^2 \, \Var_{n-1}[\phi_n(X_N)]  /  2\big\} + o_{\PP}(1)$ and the proof is complete once it is shown that 
\begin{align*}
\Var_{n-1}[\phi_n(X_N)]
&=
\sum_{i=1}^N G_{n-1,N}(x_{n-1}^i) M_{n,N}(\phi_n^2)(x_{n-1}^i) \; / \; \sum_{i=1}^N G_{n-1,N}(x_{n-1}^i)
\\
&\qquad -
\Big\{ \sum_{i=1}^N G_{n-1,N}(x_{n-1}^i) M_{n,N}(\phi_n)(x_{n-1}^i) \; / \; \sum_{i=1}^N G_{n-1,N}(x_{n-1}^i) \Big\}^2\\
&=
\eta^N_{n-1}\big[Q_{n-1, \eta^N_{n-1}(\xi_{n})}\phi_n^2 \big] \; / \; \eta^N_{n-1}\big[ G_{n-1, \eta^N_{n-1}(\xi_{n})} \big]
\\
&\qquad -
\Big\{ \eta^N_{n-1}\big[Q_{n-1, \eta^N_{n-1}(\xi_{n})}\phi_n \big] \; / \; \eta^N_{n-1}\big[ G_{n-1, \eta^N_{n-1}(\xi_{n})} \big] \Big\}^2
\end{align*}
converges in probability to $\Sigma_{\eta_n}(\phi_n)$. 
By Assumption \ref{hyp:1},  functions 
$(x,\xi) \mapsto G_{n-1, \xi}(x)$, 
$(x,\xi) \mapsto Q_{n, \xi} \phi_n(x)$, 
$(x,\xi) \mapsto Q_{n, \xi} \phi_n^2(x)$ are bounded and continuous at $\xi=\eta_{n-1}(\xi_{n})$ uniformly on $E_{n-1}$. By Corollary \ref{cor:wlln}, $\eta^N_{n-1}(\xi_{n})$ converges in probability to $\eta_{n-1}(\xi_{n})$;
by Theorem \ref{theo:wlln} and Slutsky's Lemma we get that $\Var_{n-1}[\phi_n(X_N)]$ converges in probability to
\begin{align*}
\eta_{n-1} [ Q_n(\phi_n^2) ] / \eta_{n-1}(G_{n}) 
-
\big( \eta_{n-1} [ Q_n(\phi_n) ] / \eta_{n-1}(G_{n}) \big)^2\ ,
\end{align*}
which is another formula for $\eta_n(\phi_n^2) - \eta_n(\phi_n)^2 = \Sigma_{\eta_n}(\phi_n)$,
as required.
\item
To prove that $\sqrt{N} \, \wtilde{B}(N) = \sqrt{N} \,  [\gamma^N_{n-1} - \gamma_{n-1}](\LL_n(\phi_n)) + o_{\mathbb{P}}(1)$ we write $\wtilde{B}(N)$ as
\begin{equation} \label{eq.decomposition.B}
\gamma^{N}_{n-1}(1) \times \eta^N_{n-1} [Q_{n,N}-Q_n](\phi_n)
\;+\;
[\gamma^N_{n-1}-\gamma_{n-1}](Q_n \phi_n)\ .
\end{equation}
Furthermore, we have 
\begin{equation} \label{eq.Q.diff.eta}
\eta^N_{n-1} [Q_{n,N}-Q_n](\phi_n) = \eta^N_{n-1}\,[\,\omega(\cdot, \eta^N_{n-1}(\xi_n))\,] \times
[\eta^N_{n-1}-\eta_{n-1}](\xi_n)
\end{equation}
with 
$
\omega(x, z) := \int_0^1 \partial_{\xi} Q_{n,\xi} \phi_n (x)|_{\xi= \eta_{n-1}(\xi_n) + \lambda(z - \eta_{n-1}(\xi_n))} \, d \lambda
$.
Under Assumption \ref{hyp:2}, function $\omega$ 
is bounded and continuous at $z = \eta_{n-1}(\xi_n)$ uniformly over $x \in E_{n-1}$. Theorem \ref{theo:wlln} applies so that $\eta^N_{n-1}\,[\,\omega(\cdot, \eta^N_{n-1}(\xi_n))\,]\rightarrow\eta_{n-1}\,[\, \omega(\cdot, \eta_{n-1}(\xi_n))\,] = \eta_{n-1}[\partial_{\xi}Q_n(\phi)]$, in probability. The induction hypothesis, Slutky's Lemma and standard manipulations yield that $\sqrt{N} \times \gamma^N_n [Q_{n,N}-Q_n](\phi_n)$ equals
\begin{equation*}
\sqrt{N} \times \eta_{n-1}\big[\partial_{\xi}Q_n(\phi) ] \times [\gamma^N_{n-1}-\gamma_{n-1}](\xi_n-\eta_{n-1}(\xi_n)) + o_{\PP}(1)\ .
\end{equation*}
It then follows from \eqref{eq.decomposition.B} that
$\sqrt{N} \, \wtilde{B}(N) = \sqrt{N} \, [\gamma^N_{n-1} - \gamma_{n-1}](\LL_n \phi_n) \;+\; o_{\PP}(1)$.
\end{itemize}
This concludes the proof of the induction steps and finishes the proof of Theorem \ref{thm.clt.unnormalised}.
\end{proof}

In the case where the summary statistics are constant, i.e. $\xi_p \equiv C\in\mathbb{R}$ for $p \geq 0$, expression \eqref{eq.asymp.var.unnormalised} reduces to the usual non-adaptive asymptotic variance as presented, for example, in \cite{delmoral}. In the special case $\phi_n \equiv 1$, one obtains the following expression for the asymptotic variance of the relative normalisation constant $\gamma^N_n(1) / \gamma_n(1)$.

\begin{cor} Assume (A\ref{hyp:1}-\ref{hyp:2}) and  let $n \ge 0$ be a non-negative integer. Then the quantity $\sqrt{N} \, \big\{ \gamma^N_n(1) / \gamma_n(1) - 1\big\}$ converges, as $N \to \infty$, to a centered Gaussian distribution with variance
\begin{equation*}
\sum_{p=0}^n 
\frac{\Var_{\eta_p}(\LL_{p,n} \, 1)}
{\prod_{k=p}^{n-1} \eta_k(G_k)^2}\ .
\end{equation*}
\end{cor}

Similarly, one can obtain a CLT for the empirical normalised measures $\eta^N_n(\phi_n)$:
%
%  CLT for normalized measure
%
\begin{theorem}
\label{theo:clt}
Assume (A\ref{hyp:1}-\ref{hyp:2}). Let $n \ge 0$, $r \ge 1$ and $\phi_n: E_n \to \bbR^{r}$ be a bounded measurable function. The sequence $\sqrt{N} \, [\eta^N_n - \eta_n](\phi_n)$ converges weakly to a centered Gaussian distribution  with covariance
\begin{equation} \label{eq.asymp.var.normalised}
\Sigma_n(\phi_n) 
:= 
\sum_{p=0}^n \frac{\gamma_p(1)^2}{\gamma_n(1)^2} \, \Sigma_{\eta_p}\big[\LL_{p,n} \big(\phi_n-\eta_n(\phi_n)\big) \big]
\end{equation}
with the linear operators $\LL_p$ for $p \geq 0$ as defined in \eqref{eq.semigroup}. The asymptotic variances satisfy
\begin{equation} \label{eq.asymp.var.normalised.rec}
\Sigma_n(\phi_n) := \Sigma_{\eta_n}(\phi_n) + \frac{\Sigma_{n-1}\big[\LL_{n} \big(\phi_n - \eta_n(\phi_n)\big) \big]}{\eta_{n-1}(G_{n-1})^2}\ .
\end{equation}
\end{theorem}
\begin{proof}
One can verify that the normalised measure $\eta^N_n$ is related to the unnormalised measure $\gamma^N_n$ through the identity (\cite[pp.~301]{delmoral})
\begin{equation*}
[\eta^N_n - \eta_n](\phi_n) = \frac{\gamma_n(1)}{\gamma_n^N(1)} \, \gamma_n^N\big[ \tfrac{1}{\gamma_n(1)} (\phi_n - \eta_n(\phi_n))\big]\ .
\end{equation*}
By Corollary \ref{cor.wlln.normalisation}, $\gamma_n(1)/ \gamma_n^N(1)$ converges in probability to 1. Slutsky's Lemma and Theorem \ref{thm.clt.unnormalised} yield that $\sqrt{N} \, [\eta^N_n - \eta_n](\phi_n)$ converges weakly to a centered Gaussian variable with variance 
$
\sum_{p=0}^n \gamma_p(1)^2 \, \Sigma_{\eta_p}[\LL_{p,n} \big( \gamma_n(1)^{-1} (\phi_n - \eta_n(\phi_n)\big)]
$,
which is just another way of writing \eqref{eq.asymp.var.normalised}. Equation \eqref{eq.asymp.var.normalised.rec} follows from the identities 
$\gamma_p(1) = \prod_{k=0}^{p-1} \eta_k(G_k)$,
$\eta_{n-1}(\LL_n \phi_n) = \eta_n(\phi_n)$.
\end{proof}

\subsection{Stability}
\label{sec.stability}
%Consider the following adaptive SMC algorithm. For any $n \geq 0$ the state space is finite, $E_n=E=\{0,1\}$, and the initial condition is uniform, $\eta_0(0)=\eta_0(1)=1/2$. The potential functions are non-adaptive and identically equal to one, $G_n \equiv 1$ for $n \geq 0$. The parametric family of Markov chain is not time dependent, $M_{n,\eta} \equiv M_\eta$, and $M_{\eta}$ has transition $M_{\eta}(0 \to 0)=1-\Psi(\eta)$, $M_{\eta}(0 \to 1)=\Psi(\eta)$, $M_{\eta}(1 \to 0)=1-\Psi(\eta)$, $M_{\eta}(1 \to 1)=\Psi(\eta)$, and the SMC algorithm exploits the identity statistics $\xi_n(0) = 1-\xi_n(1)=0$ for adapting the Markov kernel. Suppose that $\Psi: [0,1] \to [0,1]$ is increasing, with $\Psi(0)=0$, $\Psi(1/2)=1/2$, $\Psi(1)=1$. Indeed, for any $n \geq 0$ the distribution $\eta_n$ is uniform, $\eta_0(0)=\eta_0(1)=1/2$. The perfect algorithm has all the nice properties that one could imagine. On the other hand, if $\Psi(1/2)>1$, the resulting adaptive SMC is not stable. \red{[AT: details to come, mainly because $0$ et $1$ are attractive points]}

We now show that in the majority of applications of interest, the asymptotic variance of the adaptive SMC algorithm is identical to the asymptotic variance of the \emph{perfect} algorithm.
%
% heredity of the asymptotic variance
%
\begin{theorem} [Stability] \label{thm.stability}
Assume (A\ref{hyp:1}-\ref{hyp:2}). Suppose further that for any index $n \geq 1$ the identity
\begin{equation} \label{eq.stability}
\eta_{n-1}(G_{n-1, \xi} M_{n,\xi}) / \eta_{n-1}(G_{n-1, \xi}) = \eta_n
\end{equation}
holds for any parameter $\xi \in \domain(\xi_n)$.
For any test function $\phi_n \in \mathcal{B}_b(E_n)$, the asymptotic variance of the adaptive SMC algorithm identified in Theorem \ref{thm.clt.unnormalised} equals the asymptotic variance of the perfect SMC algorithm.
\end{theorem}
\begin{proof}
Formula \eqref{eq.semigroup} shows that it suffices to prove that the term $\eta_{n-1}(\partial_{\eta} Q_n \phi_n )$ vanishes. By differentiation under the integral sign, it is enough to prove that the mapping  $\xi \mapsto \eta_{n-1}(Q_{n, \xi} \phi_n)$ is constant on $\domain(\xi_n)$. Indeed, it follows from \eqref{eq.stability} that $\eta_{n-1}(Q_{n, \xi} \phi_n) = \eta_n(\phi_n)$ for any $\xi \in \domain(\xi_n)$, concluding the proof of Theorem \ref{thm.stability}.
\end{proof}

Theorem \ref{thm.stability} applies for instance to the sequential Bayesian parameter inference context discussed in Section \ref{sec:seq_bi} and to the filtering setting of Section \ref{sec:ex_filt}.
%\red{[any statistically interesting example where Equation \ref{eq.stability} does not hold?]}
A consequence of Theorem \ref{thm.stability} is that standard behaviours for the asymptotic variance of the \emph{perfect} SMC algorithm, such as linear growth of the asymptotic variance of $\sqrt{N} \, \big( \gamma^N_n(1) / \gamma_n(1) - 1\big)$, are inherited by the adaptive SMC algorithm.

 %
 % ADAPTIVE TEMPERING SECTION
 %
\section{Adaptive Tempering}\label{sec:annealed}
We now look at the scenario when one uses the information in the evolving particle population 
to adapt a sequence of distributions by means of a tempering parameter $\beta \in (0,1)$. 

\subsection{Algorithmic Set-Up}
In many situations in Bayesian inference one seeks to sample from a distribution $\pi$ on a set $E$ of the form
\begin{equation*}
\pi(dx) = \tfrac{1}{Z} \, e^{- \beta_{*} \, V(x)} \, m(dx)
\end{equation*}
where $Z$ is a normalisation constant, $m(dx)$ a dominating measure on the set $E$ and $V:E \to \bbR$ a potential.  Coefficient $\beta_{*}\in \bbR$ can be thought of as an inverse temperature parameter.
A frequently invoked algorithm involves forming a sequence of `tempered' probability distributions
\begin{equation*}
\eta_n(dx)  = \frac{1}{Z(\beta_n)} \, e^{-\beta_{n} V(x)} \, m(dx)
\end{equation*}
for inverse temperatures
$\beta_0 \leq \ldots \leq \beta_{n-1} \leq \beta_n\leq \cdots \leq  \beta_{n_*}=\beta_*$; in many applications $\beta_*=1$. 
The associated unnormalised measures are
\begin{equation*}
\gamma_n(dx)  = e^{-\beta_{n} V(x)} \, m(dx)\ .
\end{equation*}
Particles are propagated from $\eta_{n-1}$ to $\eta_{n}$ through a Markov kernel $M_n
$ that preserves $\eta_{n}$. In other words, the algorithm corresponds to the SMC approach discussed in Section \ref{sec:algo} with potentials
\begin{equation*}
G_n(x) 
%\equiv G_{n,\rho_{n}, \rho_{n-1}}(x) 
= e^{-\Delta_{n} \, V(x)}\ , 
\quad \Delta_n := \beta_{n+1}-\beta_n\  ,
\end{equation*}
and Markov kernels $M_n$ satisfying $\eta_n M_{n} = \eta_n$. For test function $\phi_n \in \mathcal{B}_b(E)$, the $N$-particle approximation of the normalised and unnormalised distribution are given in \eqref{eq.approx.normalized.measure}, \eqref{eq.approx.unnormalized.measure}.
To be consistent with the notations introduced in Section \ref{sec.algo1.CLT},
note that the normalisation constants also read as $Z(\beta_n) = \gamma_n(1)$ and $Z=Z(\beta_*)=\gamma_{n^*}(1)$. 
In most scenarios of practical interest, it can be difficult or even  undesirable to decide \emph{a-priori} upon the annealing sequence $\{\beta_n\}_{n=0}^{n_*}$.
Indeed, if the chosen sequence features big gaps, one may reach the terminal temperature rapidly, the variance of the weights being potentially very large due to large discrepancies between consecutive elements of the bridging sequence of probability distributions. Alternatively, if the gaps between the annealing parameters are too small, the variance of the final weights can be very small; this comes at the price of needlessly wasting a lot of computation time. Knowing what constitutes `big' or `small' with regards to the temperature gaps can be very-problem specific. Thus, an automated procedure for determining the annealing sequence is of great practical importance. In this section we  investigate the asymptotic properties  of an algorithm where the temperatures,  as well as statistics of the MCMC kernel, are  determined empirically by the evolving population of particles.

A partial analysis of the algorithm to be described can be found in \cite{giraud}. However, the way in which the annealing sequence is determined in that work does not correspond to one typically used in the literature. In addition, the authors assume that the perfect MCMC kernels are used at each time step, whereas we do
not assume so. It should also be noted, however, that the analysis in \cite{giraud} is non-asymptotic.% while we only focus on proving a WLLN and a CLT. 

The adaptive version of the above described  algorithm constructs the (random) temperatures sequence $\{ \beta_p^{N} \}_{p \geq 0}$ `on the fly'  as follows.
Once a proportion $\alpha\in(0,1)$ has been specified, the random tempering sequence is determined through the recursive equation 
\begin{equation} \label{eq.beta.recursion}
%\rho_n^N = \inf\Big\{t\in[\rho_{n-1}^N,1]:\frac{\eta_{n}^N(\kappa^{t-\rho_{n-1}^N})^2}{\eta_{n}^N(\kappa^{2(t-\rho_{n-1}^N)})}\leq \alpha\Big\}
\beta^N_{n+1} = 
\inf \Big\{ \beta^N_n < \beta \leq \beta_*
\; : \;
\ess(\eta^N_{n}, e^{-(\beta-\beta^N_n) \, V}) = \alpha
\Big\}
\end{equation}
initialized at a prescribed value $\beta_0$ typically chosen so that the distribution $\eta_0$ is easy to sample from.
% and where $\eta^N_n$ is the $N$-particle approximation of the distribution $\eta_n$. 
For completeness, we use the convention that $\inf \varnothing =\beta_{*}$. 
In the above displayed equation, we have used the $\ess$ functional defined for a measure $\eta$ on the set $E$ and a weight function $\omega: E \to (0, \infty)$ by
\begin{equation*}
\ess(\eta, \omega) := \eta(\omega)^2 / \eta(\omega^2)\ .
\end{equation*}
The following lemma guaranties that under mild assumptions the effective sample size functional $\beta \mapsto \ess(\eta_p, e^{-(\beta-\beta_n) \, V})$ is continuous and decreasing so that \eqref{eq.beta.recursion} is well-defined and the inverse temperature $\beta_{n+1}$ can be efficiently computed by a standard bisection method.
%
%  ESS IS DECREASING
%
\begin{lem} \label{lem.ess.decreasing}
Let $\eta$ be a finite measure on the set $E$ and $V:E \to \bbR$ be a bounded potential. Then, the function $\lambda \mapsto \ess(\eta, e^{-\lambda \, V})$ 
is continuous and decreasing on $[0,\infty)$. Furthermore, if $\PP\,[\,V(X) \neq V(Y)\,] > 0$ for $X,Y$ independent and distributed according to $\eta$, the function is strictly decreasing.
\end{lem}
\begin{proof}
We treat the case where $\PP\,[\,V(X) \neq V(Y)\,] > 0$, the case $\PP\,[\,V(X) \neq V(Y)\,] = 0$ being trivial. 
Let $X$ and $Y$ be two independent random variables  distributed according to $\eta$. 
The dominated convergence theorem shows that the function $\lambda \mapsto \ess(\eta, e^{-\lambda \, V})$ is continuous, with a continuous derivative. Standard manipulations show that the derivative is strictly negative if 
$\eta( V \, e^{-\lambda V}) \, \eta( e^{-2\lambda V}) >
\eta( e^{-\lambda V} ) \, \eta( V \, e^{-2 \lambda V})$, which is equivalent to the condition
\begin{equation*}
\E\,\Big[\, e^{-\lambda \{ V(X)+V(Y) \} } \times \Big\{ V(X)-V(Y) \Big\} \times \Big\{ e^{-\lambda V(X)} - e^{-\lambda V(Y)}\Big\}\,\Big] < 0\ .
\end{equation*}
This last condition is satisfied since for any $x,y \in \bbR$ and any $\lambda > 0$ we have the inequality $\{ V(x)-V(y)\} \{e^{-\lambda \, V(x)} - e^{- \lambda \, V(y)}\} < 0$, with strict inequality for $x \neq y$.
\end{proof}
We will assume that the sequence of temperatures $\{ \beta_n \}_{n \geq 0}$ and $\{ \beta^N_n\}_{n \geq 0}$ are defined for \emph{any} index $n \geq 0$,
using the convention that the first time that the parameter $\beta^N_n$ reaches the level $\beta_*$, which is random for the practical algorithm, the algorithm still goes on with fixed inverse temperatures equal to $\beta_*$. 
Under this convention, we can carry out an asymptotic  analysis using an induction argument. Ideally one would like to prove asymptotic consistency (and a CLT) for the empirical measure 
at the random termination time of the practical algorithm; we do not do this, due to the additional technical challenge that it poses.
We believe that the result to be proven still provides a very satisfying theoretical justification for the practical adaptive algorithm. 
We assume from now on that for the perfect algorithm the sequence of inverse temperatures is given by the limiting analogue of 
\eqref{eq.beta.recursion},
\begin{equation}
\label{eq:hitting}
\beta_{n+1} = 
\inf \Big\{ \beta_n < \beta \leq \beta_*
\; : \;
\ess(\eta_{n}, e^{-(\beta-\beta_n) \, V}) = \alpha
\Big\}\ .
\end{equation}

We will show in the next section that under mild assumptions the adaptive version $\beta^N_n$ converges in probability towards $\beta_n$. For statistics $\xi_{n+1}:E \to \bbR^d$ we set
\begin{equation*}
\theta^N_{n} =\big( \beta^N_n, \beta^N_{n+1},\eta^N_{n}(\xi_{n+1})^{\top} \big)^{\top}
\end{equation*}
and denote by $\theta_n$ its limiting value. At time $n$, for a particle system $\{x_n^i\}_{i=1}^N$ and associated empirical distribution $\eta^N_n$ targeting the distribution $\eta_n$, the next inverse temperature $\beta^N_{n+1}$ is computed according to \eqref{eq.beta.recursion}; the particle system is re-sampled according to a multinomial scheme with weights 
\begin{equation*}
G_{n,N}(x) := e^{-\Delta^N_{n} \, V(x)}\ ; 
\quad
\Delta^N_n = \beta^N_{n+1} - \beta^N_{n}\ , 
\end{equation*}
and then evolves via a Markov kernel $M_{n+1,N} \equiv M_{n+1, \eta^N_n(\xi_{n+1}), \beta^N_{n+1}}$
that preserves the preserves $Z(\beta^N_{n+1})^{-1} \, e^{-\beta^N_{n+1} \, V} \, m(dx)$. %that might depend on the values of $\beta^N_{n+1}$ and $\eta^N_n(\xi_{n+1})$. 
Similarly to Section \ref{sec:algo1}, we will make use of the operator
\begin{equation*}
Q_{n,N}(x,dy) \equiv
G_{n-1,N}(x) \, M_{n,N}(x,dy)
\end{equation*}
and its limiting analogue $Q_n$. With these notations, note that Equation \eqref{eq:d1} holds. To emphasise the dependencies upon the parameter $\theta=(\beta_1, \beta_2, \eta)$, we will sometimes use the expression $Q_{n,\theta} = G_{n,\theta}(x) \, M_{n,\eta,\beta_2}(x,dy)$ with $G_{n,\theta} = e^{-(\beta_2-\beta_1) \, V} = e^{- \Delta \, V}$ and $\Delta = \beta_2 - \beta_1$. For notational convenience, we sometimes write $\partial_{\Delta}$ when the meaning is clear. For example, by differentiation under the integral sign, the quantity $\partial_{\Delta} \eta_n(G_n)$ also equals $- \eta_n( V \, G_n)$. Unless otherwise stated, the derivative $\partial_\theta$ is evaluated at the limiting parameter $\theta_n = (\beta_n, \beta_{n+1}, \eta_n(\xi_{n+1}))$. 

\subsection{Assumptions}
We define $\domain(\beta)= \{ (\beta_1, \beta_2) \in [\beta_0, \beta_*]^2 \; 
; \; \beta_1 \leq \beta_2 \}$. By $\domain(\xi_p) \subset \mathbb{R}^d$ we denote a convex set that contains the range of the statistic $\xi_p:E_{p-1} \to \mathbb{R}^{d}$.
The results to be presented in the next section make use of the following hypotheses.
\begin{hypA}\label{hyp:anneal1}
The potential $V$ is bounded on the set $E$.
For each $n \geq 0$ the function $(x,\theta) \mapsto G_{n,\theta}(x)$ is bounded and continuous at $\theta_n= \big(\beta_n, \beta_{n+1}, \eta_{n}(\xi_{n+1}) \big)$  uniformly on $E$. The statistic $\xi_{n}:E \to \bbR^d$ is bounded. For any bounded Borel test function $\phi_{n}:E \to \bbR^r$, the function $(x,\theta) \mapsto Q_{n, \theta} \phi_{n} (x)$ is bounded and continuous at $\theta =\theta_{n-1}$ uniformly on $E$.
\end{hypA}

\begin{hypA}\label{hyp:anneal2}
For each $n \geq 1$, $r\ge 1$ and bounded Borel test function $\phi_{n}:E \to \bbR^r$ the function 
$(x,\theta) \mapsto \partial_{\theta} Q_{n,\theta} \, \phi_{n}(x)$ is well defined, bounded and continuous at $\theta=\theta_{n-1}$ uniformly on $E$.
\end{hypA}
These conditions could be relaxed at the cost of considerable technical complications in the proofs.

\subsection{Weak Law of Large Numbers}

In this section we prove that the consistency results of Section \ref{sec.algo1.WLLN} also hold in the adaptive annealing setting. To do so, we prove that for any index $n \geq 0$ the empirical inverse temperature parameter $\beta^N_n$ converges in probability towards $\beta_n$.
%
%  THM: WLLN annealed
%
\begin{theorem}[WLLN]\label{theo:wlln_aneal}
Assume (A\ref{hyp:anneal1}). For any $n \geq 0$, the empirical inverse temperature $\beta^N_n$ converges in probability to $\beta_n$ as $N \to \infty$. Also, 
let $\mathsf{V}$ be a Polish space and $\{V_N\}_{N \geq 0}$ a sequence of $\mathsf{V}$-valued random variables that converges in probability to $\mathsf{v} \in \mathsf{V}$. Let $r \ge 1$ and $\phi_n: E \times \mathsf{V} \to \bbR^{r}$ a bounded function continuous at $\mathsf{v} \in \mathsf{V}$ uniformly on $E$. Then, the following limit holds in probability
\begin{equation*}
\lim_{N \to \infty} \eta_n^N[\phi_n(\cdot, V_N)]
=
\eta_n[\phi_n(\cdot, \mathsf{v})]\ .
\end{equation*}
\end{theorem}
%
%  COROLLARY: WLLN annealed
%
\begin{cor}[WLLN]
\label{cor:wlln_aneal}
Assume (A\ref{hyp:anneal1}). Let $n \ge 0$, $r \ge 1$ and $\phi_n: E \to \bbR^{r}$ be  a bounded measurable function. The following limit holds in probability,
$\lim_{N \to \infty} \eta_n^N(\phi_n)
=
\eta_n(\phi_n)$.
\end{cor}
%
%  COROLLARY: WLLN annealed noralisation constant
%
\begin{cor}
\label{cor.wlln_aneal.normalisation}
Assume (A\ref{hyp:anneal1}). Let $n \ge 0$, $r \ge 1$ and $\phi_n: E \to \bbR^{d}$ a bounded measurable function. The following limit holds in probability,
$\lim_{N \to \infty} \gamma_n^N(\phi_n)
=
\gamma_n(\phi_n)$.
\end{cor}
\begin{proof}[Proof of Theorem \ref{theo:wlln_aneal}]
Clearly, it suffices tp concentrate on the  case $r=1$.
We prove by induction on the rank $n \geq 0$ that $\beta^N_{n}$ converges in probability to $\beta_{n}$ and for any test function $\phi: E \times \mathsf{V} \to \bbR$ bounded and continuous at $\mathsf{v} \in \mathsf{V}$ uniformly on $E$ 
that
$[\eta_n^N-\eta_n](\phi) \prob 0$. 
The initial  case $n=0$ is a direct consequence of WLLN for i.i.d.~random variables and Definition \ref{def.unif.cont}. 
We assume the result at rank $n-1$ and proceed to the induction step.
\begin{itemize}
\item We first focus on proving that $\beta_{n}^{N}$ converges in probability to $\beta_{n}$.
Note that $\beta^N_{n}$ can also be expressed as
\begin{equation*}
\beta^N_{n} := \inf \Big\{ \beta \in [\beta_0, \beta_*] \; : \; 
\frac{\zeta_{1,n-1}^N(\beta)}{\zeta_{2,n-1}^N(\beta)} \leq \alpha \Big\}
\end{equation*}
with
$\zeta_{1,n-1}^N(\beta) = \eta_{n-1}^N[e^{-\max(0, \beta-\beta^N_{n-1}) \, V}]^2$
and
$\zeta_{2,n-1}^N(\beta) = \eta_{n-1}^N[e^{-2\max(0, \beta-\beta^N_{n-1}) \, V}]$. 
Indeed, the limiting temperature $\beta_{n}$ can also be expressed as 
\begin{equation*}
\beta_{n} := \inf \Big\{ \beta \in [\beta_0, \beta_*] \; : \; 
\frac{\zeta_{1,n-1}(\beta)}{\zeta_{2,n-1}(\beta)} \leq \alpha \Big\}
\end{equation*}
where $\zeta_{1,n-1}(\beta)$ and $\zeta_{2,n-1}(\beta)$ are the limiting values of 
$\zeta^N_{1,n-1}(\beta)$ and $\zeta^N_{2,n-1}(\beta)$. The dominated convergence theorem shows that  the paths 
$\beta \mapsto \zeta_{1,n-1}^N(\beta) / \zeta_{2,n-1}^N(\beta)$ and 
$\beta \mapsto \zeta_{1,n-1}(\beta) / \zeta_{2,n-1}(\beta)$ are continuous; it thus suffices to prove that 
the limit
\begin{equation} \label{eq.uniform.cv.temperature}
\lim_{N \to \infty} \; 
\big\| \zeta_{1,n-1}^N(\beta) / \zeta_{2,n-1}^N(\beta)-\zeta_{1,n-1}(\beta) / \zeta_{2,n-1}(\beta)\big\|_{\infty, [\beta_0, \beta_*]}
= 0
\end{equation}
holds in probability. Lemma \ref{lem.ess.decreasing} shows that the function $\beta \mapsto \zeta_{i,n-1}^N(\beta)$ is decreasing on $[\beta_0, \beta_*]$ for any $1 \leq i \leq 2$ and $n,N \geq 1$; by standard arguments, for proving \eqref{eq.uniform.cv.temperature} it suffices to show that for any fixed inverse temperature $\beta \in [\beta_0, \beta_*]$ the difference $\zeta_{1,n-1}^N(\beta) / \zeta_{2,n-1}^N(\beta)-\zeta_{1,n-1}(\beta) / \zeta_{2,n-1}(\beta)$ converges to zero in probability. Indeed, one can focus on proving that $\zeta^N_{i,n-1}(\beta)$ converges in probability to $\zeta_{i,n-1}(\beta)$ for $i \in \{1,2\}$. We present the proof for $i=2$, the case $i=1$ being entirely similar.
\begin{itemize}
\item For the case $\beta < \beta_{n-1}$, the induction hypothesis shows that $\beta^N_{n-1}$ converges in probability to $\beta_{n-1}$. Since $\zeta^N_{2,n-1}(\beta)=1=\zeta_{2,n-1}(\beta)$ for $\beta \leq \min(\beta^N_{n-1}, \beta_{n-1})$, the conclusion follows.
\item The case $\beta \geq \beta_{n-1}$ follows from the convergence in probability of $\beta^N_{n-1}$ to $\beta_{n-1}$ and $\eta^N_{n-1}(e^{ -(\beta-\beta_{n-1}) \, V})$ to $\eta_{n-1}(e^{ -(\beta-\beta_{n-1}) \, V})$.
\end{itemize}
\item
To prove that $\eta_n^N[\phi_n(\cdot, V_N)]$ converges in probability towards $\eta_n[\phi_n(\cdot, \mathsf{v})]$, because of the convergence in probability of $\beta^N_{n}$ to $\beta_{n}$, of $\eta^N_{n-1}(\xi_n)$ to $\eta_{n-1}(\xi_n)$ and of $V_n$ to $\mathsf{v}$, one can use exactly the same approach as the one in the proof of Theorem \ref{theo:wlln}.
\end{itemize}
\end{proof}

\subsection{Central Limit Theorem}
In this section we extend the fluctuation analysis of Section \ref{sec.algo1.CLT} to the adaptive annealing setting. We prove that for a test function $\phi_n$ the empirical quantity $\gamma^N_n(\phi_n)$ converges at $N^{-1/2}$-rate towards its limiting value $\gamma_n(\phi_n)$; we give explicit recursive expressions for the asymptotic variances.
It is noted that results for $\eta_n^N(\phi_n)$ may also be proved as in Section \ref{sec.algo1.CLT}, but are omitted for brevity.
Before stating the main result of this section, several notations need to be introduced. For any $n \geq 0$ and test function $\phi_n: E \to \bbR^r$ we consider the extension operator $\ext_n$ that maps the test function $\phi_n$ to the function $\ext_n(\phi_n): E \to \bbR^{r+2}$ defined by
\begin{equation*}
\ext_n(\phi) := \Big(G_n - \eta_{n}(G_n), \, G_n^2- \eta_{n}(G^2_n), \, \phi_n \Big)^\top.
\end{equation*}
The linear operator $\A_n$ maps the bounded Borel function $\phi_n: E \to \bbR^r$ to the rectangular $(r+1)\times (r+3)$ matrix $\A_n(\phi_n)$  %\in \mat_{r+1,r+3}(\bbR)$ 
defined by
$[\A_n(\phi_n)]_{1,1}=1$, $[\A_n(\phi)]_{1,[4:r+3]} = 0_{1 \times r}$, $[\A_n(\phi_n)]_{[2:r+1],[4:r+3]} = I_{r \times r}$ and
\begin{align*}
&[\A_n(\phi_n)]_{1,2}= -2\gamma_{n-1}^{-1}(1) \, \frac{ \eta_{n-1}(G_{n-1})}{\eta_{n-1}(G^2_{n-1})}\cdot \Big\{ \partial_\Delta \Big[
\frac{\eta_{n-1}(G_{n-1})^2}{\eta_{n-1}(G_{n-1}^2)} \Big] \Big\}^{-1} \ ; \\
&[\A_n(\phi_n)]_{1,3}  =   
\gamma_{n-1}^{-1}(1) \, \frac{\eta_{n-1}(G_{n-1})^2}{\eta_{n-1}(G_{n-1}^2)^2} \cdot \Big\{ \partial_\Delta 
\Big[ \frac{\eta_{n-1}(G_{n-1})^2}{\eta_{n-1}(G_{n-1}^2)} \Big] \Big\}^{-1}\  ;\\
&[\A_n(\phi_n)]_{2:r+1,1}  = \big(\partial_{\beta_{n-1}}+\partial_{\beta_{n}}\big) \, \eta_{n-1}(Q_n \phi_n)\  ; \\
&[\A_n(\phi_n)]_{2:r+1,2}= \gamma_{n-1}(1) \, \eta_{n-1}[\partial_{\beta_{n}} Q_n \phi_n] \times [\A_n(\phi_n)]_{1,2}\ ; \\
&[\A_n(\phi_n)]_{2:r+1,3}= \gamma_{n-1}(1) \, \eta_{n-1}[\partial_{\beta_{n}} Q_n \phi_n] \times [\A_n(\phi_n)]_{1,3}\  .
&\end{align*}

%
%  THM: CLT annealing
%
\begin{theorem}[CLT]\label{theo:clt2}
Assume (A\ref{hyp:anneal1})-(A\ref{hyp:anneal2}). 
Let $n \ge 0$, $r \ge 1$ and $\phi_n: E_n \to \bbR^{r}$ be a bounded measurable function. The sequence $\sqrt{N} \, \big( \beta^N_n - \beta_n , [\gamma^N_n - \gamma_n](\phi_n) \big)^\top$ converges weakly to a centred Gaussian distribution  with covariance
\begin{equation} \label{eq.recursion.cov}
\Sigma_n(\phi_n)
\;=\;
\A_n(\phi_n) \cdot \Sigma_{n-1}\big( \ext_{n-1}( Q_n \phi_n ) \big) \cdot \A_n(\phi_n)^\top + \gamma^2_n(1) \, \wtilde{\Sigma}_{\eta_n}(\phi_n)
\end{equation}
where $\wtilde{\Sigma}_{\eta_n}(\phi_n)$ is the covariance matrix of the function $\big(0,\phi_n \big)^\top$ under $\eta_n$.
\end{theorem}
\begin{proof}
The proof follows closely the one of Theorem \ref{thm.clt.unnormalised}. For the reader's convenience, we only  highlight the differences. The proof proceeds by induction, the case $n=0$ directly following from the CLT for i.i.d random variables. For proving the induction step, assuming that the result holds at rank $n-1$, it suffices to prove that 
\begin{equation} \label{eq.conditional.rec}
\E_{n-1}\left(\begin{array}{c}
\beta^N_n - \beta_n \\
\big[ \gamma^N_n - \gamma_n\big](\phi_n)
\end{array} \right) 
\; = \;
\A_{n,N}(\phi_n)\,
\left( \begin{array}{c}  
\beta^N_{n-1} - \beta_{n-1} \\
\big[ \gamma^N_n - \gamma_n\big]  \big( \ext[Q_n \phi_n] \big)
\end{array}\right),
\end{equation}
with $\A_{n,N}(\phi_n) \in \mat_{r+1,r+3}(\bbR)$ converging in probability to $\A_n(\phi_n)$, and that for any vector $t \in \bbR^r$ the following limit holds in probability
\begin{equation*}
\lim_{N \to \infty}
\E\Big[ \exp \Big\{ i \, t \, \sqrt{N} \, C(N) \Big]
\;=\;
\exp\big\{ -\gamma^2_n(1) \, \bra{t, \Sigma_{\eta_n}(\phi_n) \, t}  / 2\big\}\ 
\end{equation*}
with 
$C(N) = (\gamma^N_n-\gamma_n)(\phi_n) - \E_{n-1}\big[ (\gamma^N_n-\gamma_n)(\phi_n) \big]$.
The proof of the above displayed equation is identical to the proof of  \eqref{eq.cv.prob.fourier} and is thus omitted. We now prove  \eqref{eq.conditional.rec}.
\begin{itemize}
\item
We first treat the term $\E_{n-1}[\beta^N_n - \beta_n] = \beta^N_n - \beta_n$.
The relation
$\ess(\eta^N_{n-1}, e^{-\Delta^N_{n-1} \, V}) = \alpha = \ess(\eta_{n-1}, e^{-\Delta_{n-1} \, V})$
can be rearranged as
\begin{align} \label{eq.ess.decomp}
\begin{aligned}
\eta_{n-1}(G_{n-1})^2 \, &\Big\{ \eta^N_{n-1}(e^{-2 \Delta^N_{n-1}V}) - \eta_{n-1}(e^{-2 \Delta_{n-1}V})\Big\} = \\
&
\eta_{n-1}(G_{n-1}^2) \, \Big\{ \eta^N_{n-1}(e^{- \Delta^N_{n-1}V})^2 - \eta_{n-1}(e^{- \Delta_{n-1}V})^2\Big\}\ .
\end{aligned}
\end{align}
Decomposing $\eta^N_{n-1}(e^{-2 \Delta^N_{n-1}V}) - \eta_{n-1}(e^{-2 \Delta_{n-1}V})$ 
as the sum of 
$\eta^N_{n-1}(e^{-2 \Delta^N_{n-1}V}) - e^{-2 \Delta_{n-1}V})$
and
$[\eta^N_{n-1} - \eta_{n-1}](G_{n-1}^2)$,
and using a similar decomposition for the difference 
$\eta^N_{n-1}(e^{- \Delta^N_{n-1}V})^2 - \eta_{n-1}(e^{- \Delta_{n-1}V})^2$,
one can exploit the boundedness of the potential $V$, Theorem \ref{theo:wlln_aneal} and 
the same approach as the one used for proving  \eqref{eq.Q.diff.eta} to obtain that 
$\eta^N_{n-1}(e^{-2 \Delta^N_{n-1}V}) - \eta_{n-1}(e^{-2 \Delta_{n-1}V})$ equals
\begin{align} \label{eq.ess.d.1}
\Big\{ \partial_{\Delta}\eta_{n-1}(G^2_{n-1}) + o_{\PP}(1) \Big\} 
\times (\Delta_{n-1}^N - \Delta_{n-1})+
[\eta^N_{n-1} - \eta_{n-1} ] (G_{n-1}^2)
\end{align}
and 
$[\eta^N_{n-1}(\kappa^{\Delta^N_{n-2}})^2-\eta_{n-2}(\kappa^{\Delta_{n-2}})^2]$ equals
\begin{align} \label{eq.ess.d.2}
\begin{aligned}
\Big\{ 2 \eta_{n-1}&(G_{n-1}) \partial_{\Delta}\eta_{n-1}(G_{n-1}) + o_{\PP}(1) \Big\} \times (\Delta_{n-2}^N - \Delta_{n-2})\\
&+
\Big\{ 2 \eta_{n-1}(G_{n-1}) + o_{\PP}(1) \Big\} \times
[\eta^N_{n-1} - \eta_{n-1} ] (G_{n-1})\ .
\end{aligned}
\end{align}
Since $(\Delta^N_{n-1}-\Delta_{n-1})$ equals $(\beta^N_{n}-\beta_{n})+(\beta^N_{n-1}-\beta_{n-1})$, Slutsky's Lemma, Equations \eqref{eq.ess.decomp}, \eqref{eq.ess.d.1}, \eqref{eq.ess.d.2} and standard algebraic manipulations yield
\begin{align} \label{eq.beta.rec}
\begin{aligned}
(\beta_n^N-\beta_n) 
&= 
[\A_{n,N}(\phi)]_{1,1} \, (\beta_{n-1}^N-\beta_{n-1}) \\
&\qquad +
[\A_{n,N}(\phi)]_{1,2} \, [\gamma_{n-1}^N-\gamma_{n-1}]
\big(G_{n-1}-\eta_{n-1}(G_{n-1}) \big)\\
&\qquad +
[\A_{n,N}(\phi)]_{1,3} \, [\gamma_{n-1}^N-\gamma_{n-1}]\big(G^2_{n-1}-\eta_{n-1}(G^2_{n-1}) \big)
\end{aligned}
\end{align}
where $[\A_{n,N}(\phi)]_{1,i}$ converges in probability to $[\A_{n,N}(\phi)]_{1,i}$ for $1 \leq i \leq 3$.
\item
To deal with the term $\E_{n-1}\big[ (\gamma^N_n - \gamma_n)(\phi_n) \big]$ we make use of the decomposition 
\begin{equation} \label{eq.gamma.rec}
\E_{n-1}\big[ (\gamma^N_n - \gamma_n)(\phi_n) \big]
=
\gamma_{n-1}^N(1) \times \eta^N_{n-1}[Q_{n,N}-Q_n](\phi_n) + [\gamma^N_{n-1}-\gamma_{n-1}](Q_n \phi_n)\ .
\end{equation}
Assumptions (A\ref{hyp:anneal1})-(A\ref{hyp:anneal2}), Theorem \ref{theo:wlln_aneal} and 
the same approach as the one used for proving  \eqref{eq.Q.diff.eta} show that the term $\eta^N_{n-1}[Q_{n,N}-Q_n](\phi_n)$ equals
\begin{equation*}
\Big\{ \eta_{n-1}[\partial_{\beta_{n-1}}Q_n \phi_n ] + o_{\PP}(1)\Big\} \, (\beta^N_{n-1}-\beta_{n-1})
+
\Big\{ \eta_{n-1}[\partial_{\beta_{n}}Q_n \phi_n] + o_{\PP}(1)\Big\} \, (\beta^N_{n}-\beta_{n})\ .
\end{equation*}
Note that there is no term involving the derivative with respect to the value of the summary statistics; indeed, this is because for any value of $\xi \in \bbR^r$ the Markov kernel $M_{n,\xi}$ preserves $\eta_{n}$ so that one can readily check that $\eta_{n-1}[ \partial_{\xi} Q_{n,\xi} \phi_n] = 0$. One can then use  \eqref{eq.beta.rec} to express $(\beta^N_{n}-\beta_{n})$ in terms of the three quantities  $(\beta^N_{n-1}-\beta_{n-1})$, $[\gamma_{n-1}^N-\gamma_{n-1}]
\big(G_{n-1}-\eta_{n-1}(G_{n-1}) \big)$ and $[\gamma_{n-1}^N-\gamma_{n-1}]
\big(G^2_{n-1}-\eta_{n-1}(G^2_{n-1}) \big)$ and obtain, via Slutsky's Lemma and  \eqref{eq.gamma.rec}, that for any coordinate $1 \leq i \leq r$,
\begin{align} \label{eq.gamma.rec.2}
\begin{aligned}
\E_{n-1}\big[ (\gamma^N_n - \gamma_n)(\phi_n) \big]_i
&= 
[\A_{n,N}(\phi)]_{i+1,1} \, (\beta_{n-1}^N-\beta_{n-1}) \\
&\qquad +
[\A_{n,N}(\phi)]_{i+1,2} \, [\gamma_{n-1}^N-\gamma_{n-1}]
\big(G_{n-1}-\eta_{n-1}(G_{n-1}) \big)\\
&\qquad +
[\A_{n,N}(\phi)]_{i+1,3} \, [\gamma_{n-1}^N-\gamma_{n-1}]\big(G^2_{n-1}-\eta_{n-1}(G^2_{n-1}) \big)\\
&\qquad +
[\gamma_{n-1}^N-\gamma_{n-1}](Q_n \phi_n)_i
\end{aligned}
\end{align}
where $[\A_{n,N}(\phi)]_{i+1,j}$ converges in probability to  $[\A_{n}(\phi)]_{i+1,j}$ for $1 \leq j \leq 3$.
\end{itemize}
Equation \eqref{eq.conditional.rec} is a simple rewriting of 
\eqref{eq.beta.rec} and \eqref{eq.gamma.rec.2}. This concludes the proof.
\end{proof}

%
% APPLICATIONS
%
\section{Applications}\label{sec:exam}

%
%  VERIFICATION ASSUMPTIONS
%
\subsection{Verifying the Assumptions}
We consider the sequential Bayesian parameter inference framework of Section \ref{sec:seq_bi}. That is, for a parameter $x \in E=\bbR^m$, observations $y_i \in \mathcal{Y}$ and prior measure with density $\eta_0(x)$ with respect to the Lebesgue measure in $\bbR^m$.
%\red{[work on that: statistics need to be bounded]}
%We consider the sequential Bayesian parameter inference framework of Section \ref{sec:seq_bi} for a scalar parameter $x \in E=\bbR$.
%The filtering context in Section \ref{sec:ex_filt} is subject of current work.
We assume the following.
\begin{hypB}\label{hyp:b1}
For each $n \geq 1$ the function $G_n(x) := \PP[y_{n+1} \mid y_{1:n},x]$ is bounded and strictly positive. The statistics $\xi_n: E \to \bbR^d$ is bounded.
\end{hypB}
\begin{hypB}\label{hyp:b2}
For each $n\geq 1$, the parametric family of Markov kernel $M_{n,\xi}$ is given by a Random-Walk-Metropolis kernel.  The proposal density $q(\cdot; \xi)$ is symmetric; for a current position $x \in E$ the proposal $y$ is such that $\PP(y-x \in du) = q(u; \xi) \, du$. We suppose that the first and second derivatives
\begin{equation*}
\xi \mapsto \nabla_{\xi} q(u;\xi) \ ; \quad
\xi \mapsto \nabla^2_{\xi} q(u;\xi) \ ,
\end{equation*}
are bounded on the range $\domain(\xi_n)$ of the adaptive statistics $\xi_n: E \to \domain(\xi_n) \subset \bbR^d$.
\end{hypB}
%\begin{hypB}
%\label{hyp:b2}
%For each $n\geq 1$, parametric family of Markov kernel $M_{n,\xi}$, with $\xi=(\xi^{(1)},\xi^{(2)})$, is a %random walk Metropolis kernel with proposal
%
%\begin{equation*}
%X' = X  + \sqrt{H(\xi^{(1)} ,\xi^{(2)})}\,Z
%\end{equation*}
% 
%for $Z\sim\mathcal{N}(0,1)$. The function $H: \mathbb{R^+}\times\mathbb{R} \to \mathbb{R}_+$ is %smooth, bounded away from zero and infinity, with first and second derivatives bounded on $%\mathbb{R}_+$. The adaptive SMC algorithm make use of the statistics $\xi_n = (x,x^2)$
%\end{hypB}
%

Assumption (B\ref{hyp:b1}) is reasonable and satisfied by many real statistical models. Similarly, it is straightforward to construct proposals verifying Assumption (B\ref{hyp:b2}); one can for example show that for a function $\sigma: \domain(\xi_n) \to \bbR_+$, bounded away from zero with bounded first and second derivatives, the Gaussian proposal density $q(u; \xi) := \exp\big\{ -u^2 / [2 \sigma^2(\xi)] \big\} / \sqrt{2 \pi \sigma^2(\xi)}$ satisfies Assumption (B\ref{hyp:b2}); multi-dimensional extensions of this settings are readily constructed.

%
% Proposition: verification assumptions
%
\begin{prop} \label{prop.verif.assump}
Assume (B\ref{hyp:b1}-\ref{hyp:b2}). The kernels $(M_{n,\cdot})_{n\geq 1}$  and potentials $(G_n)_{n\geq 0}$ satisfy Assumptions (A\ref{hyp:1}-\ref{hyp:2}).
\end{prop}
\begin{proof}
By assumption, the potentials $\{ G_n \}_{n \geq 0}$ are bounded and strictly positive and the statistics $\xi_n: E \to \bbR^d$ are bounded.
To verify that Assumptions (A\ref{hyp:1}-\ref{hyp:2}) are satisfied, it suffices to prove that for any test 
function $\phi \in \mathcal{B}_b(E)$, the first and second derivatives of  
$(x,\xi) \mapsto M_{n,\xi} \phi (x)$ exist and are uniformly bounded.
The Metropolis-Hastings accept-reject 
ratio of the proposal $x \mapsto x+u$ is $r(x,u) := \min\big\{ 1, \big( \PP[y_{1:n} \mid x+u] \, \eta_0(x+u)  \big) \, / \, \big( \PP[y_{1:n} \mid x] \, \eta_0(x)  \big) \big\}$
and we have
%
%\begin{equation*}
$M_{n,\xi}(\phi)(x) 
=
\phi(x) + 
\int_{\bbR^m} \big[\phi(x+u) - \phi(x) \big] \, r(x,u) \, q(u; \xi)  \, du$. 
%\end{equation*}
Differentiation under the integral sign yields
\begin{align*}
&\nabla_{\xi} M_{n,\xi}(\phi)(x) 
= \int \big[\phi(x+u) - \phi(x) \big] \, r(x,u) \, 
\nabla_{\xi} q(u;\xi) \, du\  , \\
&\nabla^2_{\xi} M_{n,\xi}(\phi)(x) 
= \int \big[\phi(x+u) - \phi(x) \big] \, r(x,u)
\nabla^2_{\xi} q(u;\xi) \, du\ ,
\end{align*}
and the conclusion follows by boundedness of the first and second derivative of $q(u;\xi)$ with respect to the parameter $\xi \in \domain(\xi_n)$.
\end{proof}

%
% NUMERICAL EXAMPLE: Navier Stokes
%
\subsection{Numerical Example}\label{sec:num_ex}
We now provide a numerical study of a high-dimensional sequential Bayesian parameter inference, as described in  Section \ref{sec:seq_bi}, applied to the Navier-Stokes model. In this section, we briefly describe the Navier-Stokes model, the associated SMC algorithm and focus on the analysis of the behavior of the method when estimating the normalising constant. The SMC method to be presented is described in detail in \cite{kantas}.  In the subsequent discussion, we highlight the algorithmic 
challenges and the usefulness of the adaptive SMC methodology when applied to such high-dimensional scenarios. This motivates theoretical results presented in Section \ref{sec:simos} where the stability properties of the SMC estimates are investigated in the regime where the dimension $d$ of the adaptive statistics is large.

%
% Model description
%
\subsubsection{Model Description}
We work with the Navier-Stokes dynamics describing the incompressible flow of a fluid in a two dimensional torus $\mathbb{T}=[0,2\pi)\times[0,2\pi)$. The time-space varying velocity field is denoted by $v(t,x):[0,\infty) \times \mathbb{T} \rightarrow \bbR^{2}$. The Newton's laws of motion yield the Navier-Stokes system of partial differential equations \cite{doer}
\begin{align}
\begin{aligned}\label{eq:NSPDE}
\partial_{t}v-\nu\Delta v+(v\cdot\nabla)\, v=f-\nabla\mathfrak{p}\ ,
\qquad 
\nabla\cdot v=0\ ,
\qquad
\int_{\mathbb{T}}v(x,\cdot) \, dx=0\ , 
\end{aligned}
\end{align}
with initial condition $v(x,0)=u(x)$. The quantity $\nu>0$ is a viscosity parameter, $\mathfrak{p}:\mathbb{T}\times[0,\infty) \rightarrow \bbR$ is the pressure field and $f:\mathbb{T}\rightarrow\bbR^{2}$ is an exogenous time-homogeneous forcing. For simplicity, we assume periodic boundary conditions. We adopt a Bayesian approach for inferring the unknown initial condition $u = u(x)$ from noisy measurements of the evolving velocity
field $v(\cdot,t)$ on a fixed grid of points $\big(x_{1},\ldots,x_{M} \big) \in \mathbb{T}$. Performing inference with this type of data is referred to as Eulerian data assimilation.
Measurements are available at time $t_j := j \times \delta$ for time increment $\delta > 0$ and index $1 \leq j \leq T$ at each fixed location $x_m \in \mathbb{T}$. We assume i.i.d Gaussian measurements error with standard deviation $\varepsilon > 0$ so that the noisy observations $y:=\big\{ y_{j,m}\}_{j,m}$ for $1 \leq j \leq T$ and $1 \leq m \leq M$ can be modelled as
\begin{equation*}
y_{j,m}=v\left(x_{m},t_j\right)+\varepsilon\,\zeta_{j,m}
\end{equation*}
for an i.i.d sequence $\zeta_{j,m}\stackrel{iid}{\sim}\Normal(0,I_2)$.
We follow the notations of \cite{kantas} and set
\begin{equation*}
\mathbb{U}=\Big\{ \left.2\pi\textrm{-}\mbox{periodic trigonometric polynomials }u:\:\mathbb{T}\rightarrow\mathbb{R}^{2}\right|\:\nabla\cdot u=0\ ,\:\int_{\mathbb{T}}u(x)dx=0\,\Big\}\ .
\end{equation*}
We use a Gaussian random field prior for the unknown initial condition; as will become apparent from the discussion to follow, it is appropriate in this setting to assume that the initial condition $u=u(x)$ belongs the closure $U$ of $\mathbb{U}$ with respect to the $\big(L^{2}(\mathbb{T}) \big)^{2}$ norm. 
The semigroup operator for the Navier-Stokes PDE is denoted by $\Psi:U \times [0, \infty) \rightarrow U$ so that the likelihood for the noisy observation $y$ reads
\begin{equation}
\label{eq:likeli}
%p(y|u) \propto \prod_{j=1}^{T}\prod_{m=1}^{M}\exp\left(-\tfrac{1}{2\varepsilon^{2}}\big(\, y_{j,m}-
%\{\Psi(u, j\delta)\}(x_m)\,\big)^{2}\right).
\ell(y;u) 
=
\exp \Big\{ -\frac{1}{2 \varepsilon^2} \sum_{j=1}^T \sum_{m=1}^M \big\|y_{j,m}-
[\Psi(u, t_j)](x_m)\,\big\|^{2} \Big\} / (2 \pi \varepsilon^2)^{MT}\ .
\end{equation}
Under periodic boundary conditions, an appropriate orthonormal basis for $U$ is comprised of the functions $\psi_{k}(x) := \big( k^{\perp} / (2\pi \, |k| \big) \, e^{ik\cdot x}$ for $k \in \mathbb{Z}^2_* := \mathbb{Z}^{2}\setminus\{(0,0)\}$ and $k^{\perp}:=(-k_{2},k_{1})^\top$, $\left|k\right| :=\sqrt{ k_1^2+k_2^2}$. The index $k$ corresponds to a bivariate frequency and the Fourier series decomposition of an element $u\in U$ reads
\begin{equation}
\label{eq:Fourier}
u(x)=\sum_{k\in\mathbb{Z}^{2}_*} \, u_{k} \, \psi_{k}(x)
\end{equation}
with Fourier coefficients $u_{k} = \langle u,\psi_{k} \rangle = \int_{\mathbb{T}} u(x) \cdot\overline{\psi}_{k}(x) \, dx$. Since the initial condition $u \in U$ is real-valued we have $\overline{u_{k}}=-u_{-k}$ and one can focus on reconstructing the frequencies in the subset
\begin{align*}
\mathbb{Z}_{\uparrow}^{2}=
\Big\{ k=(k_{1},k_{2})\in\mathbb{Z}^{2}_* \, : \,  [k_{1}+k_{2}>0] \; \textrm{or} \;  [k_{1} = -k_{2} >0]\Big\}.
\end{align*}
We adopt a Bayesian framework and assume a centred Gaussian random field prior $\eta_0$ on the unknown initial condition
\begin{equation}\label{eq:prior}
\eta_{0}=\Normal(0,\beta^{2}A^{-\alpha})
\end{equation}
with hyper-parameters $\alpha,\beta$ affecting the roughness and magnitude of the initial vector field. In \eqref{eq:prior}, $A=-P\Delta$ denotes the Stokes operator where $\Delta = \big(\partial^2_{x_1}+\partial^2_{x_2},\,\partial^2_{x_1}+\partial^2_{x_2} \big)$ is the usual Laplacian and $P: \big(L^{2}(\mathbb{T}) \big)^{2}\rightarrow U$ is the Leray-Helmholtz orthogonal projector that maps a field to its divergence-free and zero-mean part. %hyper-parameters $\alpha,\beta$ affect the roughness and
%magnitude of the initial vector field. 
%A simple understanding of the prior can be obtained through the Karhunen-Lo\'{e}ve representation:
%
%\begin{equation}
%\label{eq:prior}
%u\sim\eta_{0}\ \Leftrightarrow\ \mbox{Re}(u_{k}),\,\mbox{Im}(u_{k})\stackrel{iid}{\sim}\mathcal{N}(0,\tfrac{1}{2}\beta^{2}|k|^{-2\alpha})\ ,\,\, k\in\mathbb{Z}_{\uparrow}^{2}\ .
%\end{equation}
%
A simple understanding of the prior distribution $\eta_0$ can be obtained through the Karhunen-Lo\'{e}ve representation; a draw from the prior distribution $\eta_0$ can be realised as the infinite sum
\begin{equation}
%\label{eq:prior}
\mathfrak{Z} = \beta \, \sum_{k\in\mathbb{Z}^{2}_*} \left| k \right|^{-\alpha} \, \xi_{k} \,\psi_{k} \; \sim \; \eta_0
\end{equation}
where variables $\{ \xi_k \}_{k \in \mathbb{Z}^2_*}$ correspond standard complex centred Gaussian random variables with $\big( \mbox{Re}(\xi_{k}), \mbox{Im}(\xi_{k}) \big)$ $\stackrel{iid}{\sim} \Normal\big(0, \frac12 \, I_2 \big)$  for $k\in\mathbb{Z}_{\uparrow}^{2}$ and $\xi_{k}=-\overline{\xi_{-k}}$ for $k \in \mathbb{Z}^{2}_* \setminus \mathbb{Z}_{\uparrow}^{2}$. In other words, \emph{a-priori}, the Fourier coefficients $u_{k}$ with $k\in\mathbb{Z}_{\uparrow}^{2}$ are assumed independent, normally distributed, with a particular rate of decay for their variances as $\left|k\right|$ increases.  Statistical inference is carried out by sampling from the  posterior probability measure $\eta$ on $U$ defined as the Gaussian change of measure
\begin{equation} \label{eq:target}
\frac{d\eta}{d\eta_{0}}(u) =  \frac{1}{Z(y)} \, \ell(y; u)
\end{equation}
for a normalisation constant $Z(y)>0$.

%
% ALGORITHMIC CHALLENGE
%
\subsubsection{Algorithmic Challenges and Adaptive SMC}
\label{sec:what}

With a slight abuse of notation we will henceforth use a single subscript to count the observations and set $y_{(j-1)M+m} \equiv y_{j,m}$.
We will apply an SMC sampler on the sequence of distributions $\{ \eta_n \}_{n=0}^{M \times T}$ defined by
\begin{equation}
\label{eq:seq}
\frac{d \eta_n}{d\eta_0}(u)  = \frac{1}{Z(y_{1:n})} \, \ell(y_{1:n}; u)
\end{equation}
for a normalisation constant $Z(y_{1:n})$ and likelihood $\ell(y_{1:n}; u)$.
Note that the state space $U$ is infinite-dimensional even though in practice, as described in \cite{kantas}, our solver truncates
the Fourier expansion \eqref{eq:Fourier} on a pre-specified window of frequencies $-k_{\max}+1 \le k_1, k_2 \le k_{\max}$ for $k_{\max}=32$. 

We now describe the MCMC mutation steps used for propagating the $N$-particle system. For a tuning parameter $\rho\in (0,1)$, a simple Markov kernel suggested in several articles (see e.g.\@ \cite{cotter} and the references therein) for target distributions that are Gaussian changes of measure of the form \eqref{eq:seq} is the following. Given the current position $u \in U$, the proposal $\widetilde{u}$ is defined as
\begin{equation}\label{eq:RWW}
\widetilde{u} = \rho \, u + (1-\rho^2)^{1/2} \, \mathfrak{Z}
\end{equation}
%
%
%\begin{equation}
%\label{eq:RWW}
%\widetilde{u}_k =\rho\, u_k+\sqrt{1-\rho^{2}}\,\mathcal{N}_2(0,\tfrac{1}{2}\beta^{2}|k|^{-2\alpha}I_2) 
%\end{equation}
with $\mathfrak{Z} \sim \eta_0$; the proposal is accepted with probability $\min\big(1, \ell(y_{1:n};\widetilde{u}) / \ell(y_{1:n};u) \big)$. Proposal \eqref{eq:RWW} preserves the prior Gaussian distribution \eqref{eq:prior} for any $\rho\in (0,1)$ and the above Markov transition is well-defined on the infinite-dimensional space $U$. It follows that the method is robust upon mesh-refinement in the sense that $\rho$ does not need to be adjusted as $k_{\max}$ increases \cite{pillai2014}. In contrast, for standard Random-Walk Metropolis 
proposals, one would have to pick a smaller step-size upon mesh-refinement; for the optimal step-size, the mixing time will typically deteriorate as $\mathcal{O}(k_{\max}^2)$, see e.g.~\cite{besk}. Still, 
proposal \eqref{eq:RWW} can be inefficient when targeting the posterior distribution $\eta$ when it differs significantly from the prior distribution $\eta_0$. Indeed, \emph{a-priori} the Fourier coefficients $u_k$ have known scales appropriately taken under consideration in \eqref{eq:RWW}; \emph{a-posteriori}, information from the data spreads non-uniformly 
on the Fourier coefficients, with more information being available for low frequencies than for high ones. Taking a glimpse into results from the execution of the adaptive SMC algorithm yet to be defined, in Figure~\ref{ex2:circle} we plot the 
fractions, as estimated by the SMC method, between posterior and prior standard deviations for the Fourier coefficient $\mathrm{Re}(u_k)$ (left panel) and  
$\mathrm{Im}(u_k)$ (right panel) over all pairs of frequencies $k=(k_1,k_2)$ with $-20\le k_1,k_2 \le 20$. In this case it is apparent that most of the information in the data concentrates on a window of frequencies around the origin; still there is a large number of variables (around $2\cdot 10^2$ in this example) which have diverse posterior standard deviations under the posterior distribution. The standard deviations of these Fourier coefficients can potentially be very different from their prior standard deviations.

%\begin{comment}
%\vspace{-6cm}
\begin{figure}[!h]
%\hspace{-1cm}
\centering \includegraphics[width=16cm]{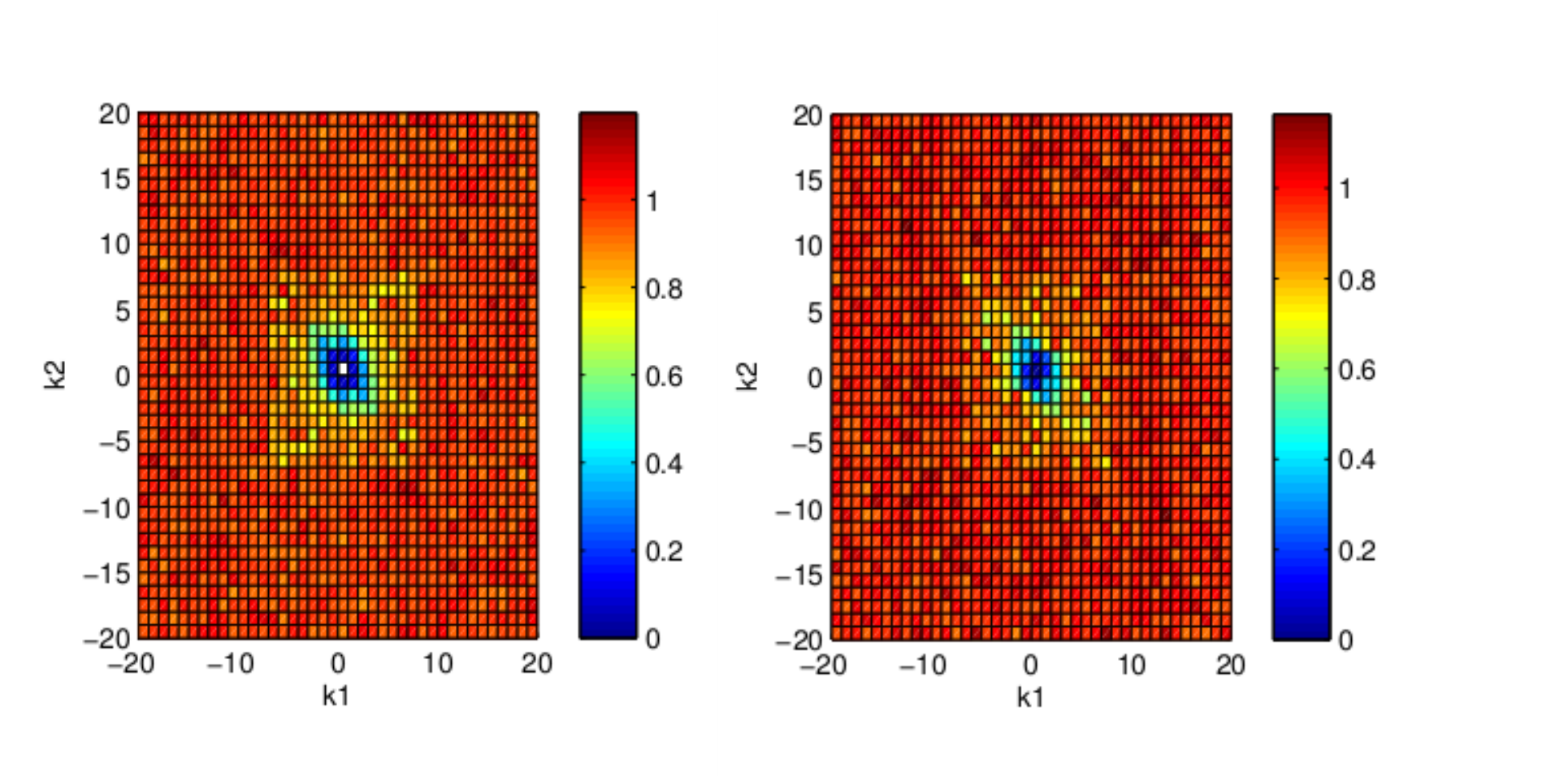}
\caption{Ratio of (estimated) posterior vs prior standard deviations for $\mathrm{Re}(u_k)$ (left panel) and  
$\mathrm{Im}(u_k)$ (right panel) over all pairs  $k=(k_1,k_2)$ with $-20\le k_1,k_2 \le 20$. 
The model here corresponds to: $\delta = 0.2$, $m=4$, $T=20$, $\alpha=2$, $\beta^2=5$, $\varepsilon^2=0.2$,
$f(x)=\nabla^{\perp}\cos((5,5)'\cdot x)$. The $m=4$ observation locations 
were at  $(0,\pi)$, $(\pi,0)$, $(0,0)$, $(\pi,\pi)$. 
Samples from the posterior were generated by applying a version of the adaptive SMC algorithm described in Section \ref{sec:what}  for $K=7$, see \cite{kantas} for full details. The `true' initial condition was sampled from the prior; data were then simulated accordingly.}
\label{ex2:circle} 
\end{figure}
%\end{comment}

The approach followed in \cite{kantas} for constructing better-mixing Markov kernels involves selecting a `window' of frequencies $\mathbf{K}=\left\{ k\in\mathbb{Z}^{2}_* \,:\, \max( k_{1}, k_{2} )\leq K \right\}$, for a user pre-specified threshold $K \geq 1$, and using the following Markov mutation steps within an SMC algorithm.
%
% improved MCMC proposal
%
\begin{itemize}
\item Use the currently available particles approximation $\{u^{i}\}_{i=1}^{N}$ of $\eta_n$ to estimate the current marginal mean and covariance $\mathfrak{m}^{N}_k$ and $\Sigma^N_k$ of the two-dimensional variable $u_k=\big(\mathrm{Re}(u_k),\mathrm{Im}(u_k) \big)$ over the window $k=(k_1, k_2) \in \mathbf{K} \cap \mathbb{Z}_{\uparrow}^{2}$,
\begin{equation*}
\mathfrak{m}^{N}_k = \tfrac{1}{N}\,\sum_{i=1}^{N}u_k^{i}
\ ; \quad 
\Sigma^N_k = \tfrac{1}{N-1} 
 \sum_{i=1}^{N}
(u_k^{i}-\mathfrak{m}^{N}_k) \otimes (u_k^{i}-\mathfrak{m}^{N}_k)\  .
\end{equation*}
For high-frequencies $k=(k_1, k_2) \in \mathbf{K}^c \cap \mathbb{Z}_{\uparrow}^{2}$, only the  information contained in the prior distribution is used and we thus set $\mathfrak{m}^{N}_k = 0$ and $\Sigma^N_k = \frac{1}{2}\, |k|^{-2\alpha} \, I_2$.
\item
For a current position $u = \sum u_k \, \psi_k$, the proposal $\widetilde{u}=\sum \widetilde{u}_k \, \psi_k$ is defined as
\begin{equation*}
%\label{eq:RWW2}
\widetilde{u}_k = \mathfrak{m}^N_{k} + \rho \, (u_{k} - \mathfrak{m}^N_{k}) + 
(1-\rho^{2})^{1/2}\, \mathfrak{Z}_k
\end{equation*}
for $k \in \mathbb{Z}_{\uparrow}^{2}$ and $\mathfrak{Z}_k \sim \Normal(0,\Sigma^N_k)$ and $\widetilde{u}_{-k} = -\overline{\widetilde{u}_{k} }$ for $\mathbb{Z}^2_* \setminus \mathbb{Z}_{\uparrow}^{2}$; this proposal is accepted with the relevant Metropolis-Hastings ratio.
\item In addition to the above adaptation at the Markov kernel, the analytical algorithm also involved an annealing step as 
described in Section \ref{sec:annealed}, whereby additional intermediate distributions were introduced, if needed, in between 
any pairs $\eta_{n-1}$, $\eta_n$. We found this to be important for avoiding weight degeneracy and getting a stable algorithm.
 As explained in Section~\ref{sec:annealed}, the choice of temperatures was determined 
on the fly, according to a minimum requirement of the effective sample size (we choose $\alpha=\tfrac{1}{3}$).
\end{itemize}

%In the problem of interest, $d=4096$ is very large and one does not know how to specify the proposal covariance matrix in the MCMC kernels, to ensure that the associated kernel will be able to explore the state-space. We note that the SMC algorithm we will implement is even more complex than that in Section \ref{sec:smc_samp} in that several intermediate probability
%measures are used in-between the $\eta_n$, but we do not include details of this; they can be found in \cite{kantas}. The proposal covariance matrix in the MCMC kernels is (essentially) found by using the empirical covariance matrix at the previous time point; the exact details are in \cite{kantas}.

It is important to note that in this Navier-Stokes setting, the regularity assumptions adopted in the theoretical parts of this article for the derivation of the asymptotic results do not apply anymore. As illustrated by this numerical analysis, the asymptotic behaviour predicted in Theorem \ref{thm.stability} is likely to hold in far more general contexts.
Figure \ref{fig:1} shows a plot of an estimate of the variance of $Z^{N}(y_{1:n}) / Z(y_{1:n})$, where  $Z^{N}(y_{1:n})$ is the $N$-particle particle approximation of normalisation constant $Z^{N}(y_{1:n})$, as a function of the amount of data $n$ for an adaptive SMC algorithm using $N=500$ particles. In this complex setting, the numerical results seem to confirm the theoretical asymptotic results of Theorem \ref{thm.stability}: the estimated asymptotic variance seems to grow linearly with $n$, as one would have expected to be true for the perfect SMC algorithm that does not use adaptation. 
This is an indication that Theorem \ref{thm.stability} is likely to hold under weaker assumptions than adopted in this article.

%\begin{comment}
\begin{figure}[h!]
\vspace{-4cm}
\centering
{\includegraphics[width=\textwidth,height=18cm]{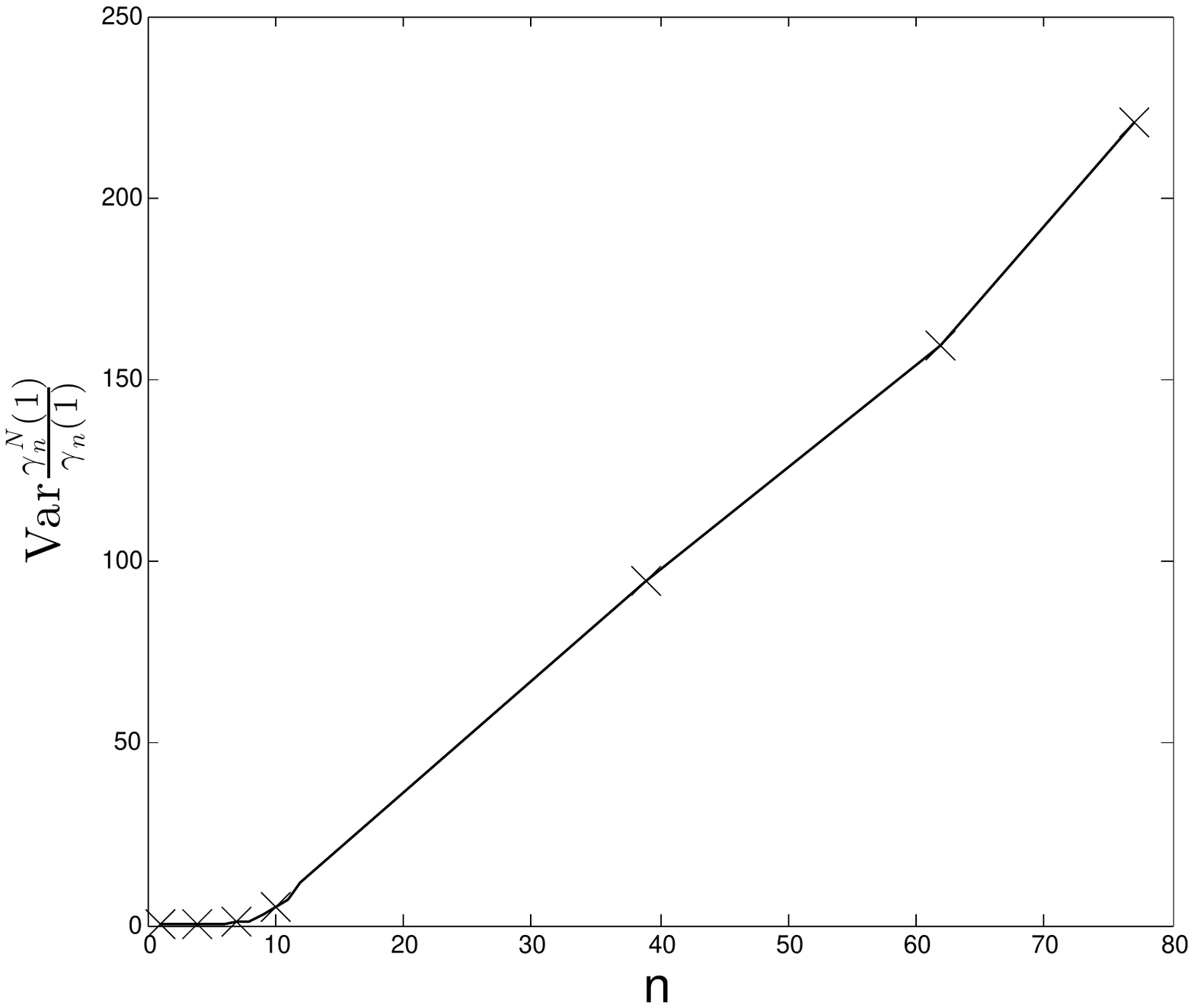}}
\vspace{-5cm}
\caption{Estimated variance for the estimate of the normalizing constant of adaptive SMC. The `true' normalizing constant was estimated from 1000 independent runs with $N=500$ and the relative variance is estimated when $N=500$ over 500 independent runs. The crosses are the estimated values of the relative variance.
%In the application of SMC, we used the adaptive MCMC step (\ref{eq:RWW2}) with $K=7$, thus the dimension of the adapted 
%statistic is $d=$ 
\label{fig:1}}
\end{figure}
%\end{comment}

\subsubsection{Algorithmic Stability in Large Scale Adaptation}\label{sec:simos}

When the dimension $d$ of the adapted statistics is large, as in the Navier-Stokes case 
(in our simulation study $d=\textbf{Card}(\mathbf{K} \cap \mathbb{Z}_{\uparrow}^{2}) \times 5\approx [(2K)^2/2] \times 5 \approx   500$) and potentially in other scenarios, 
it is certainly of interest to quantify the effect of the dimensionality $d$ of the adaptive statistics on the overall accuracy of the SMC estimators.
We will make a first modest attempt to shed some light on this issue via the consideration of a very simple modelling structure motivated by the Navier-Stokes example and allowing for some simple calculations.

%\emph{Algorithmic Set-Up:} 
For each $n \geq 1$ we assume a product form Gaussian target on $E_n=\bbR^{\infty}$,
\begin{equation*}
\eta_{n} =  \bigotimes_{j=1}^{\infty} \mathcal{N}(0,\sigma_j^{2})\ ,
\end{equation*}
for a given sequence of variances $\{\sigma_j^2\}_{j=1}^{\infty}$ that does not depend on the index $n \geq 1$.  This represents an optimistic case where the incremental weights $G_{n}(x)$ are small enough to be irrelevant for the study of the influence of the dimension $d$; we set $G_n(x)\equiv 1$.
It is assumed that the SMC method has worked well up-to time $(n-1)$ and has produced a collection of i.i.d.\@ samples $\{x_{n-1}^{i}\}_{i=1}^{N}$ from $\eta_{n-1}$. For the mutation step, we consider an adaptive Metropolis-Hastings Markov kernel $M_{n,\xi}$ preserving $\eta_{n}$ that proposes, when the current position is $x \in \bbR^{\infty}$, a new position $\wtilde{x} \in \bbR^{\infty}$ distributed as
\begin{equation} \label{eq:M} 
\begin{aligned}
\widetilde{x}_j &= \rho \, x_j + (1-\rho^2)^{1/2} \, \Normal(0, \widehat{\sigma}^2_{j})\ , 
\quad \textrm{for} \quad 1 \leq j \leq d \ ,\\
\widetilde{x}_j &= \rho \, x_j + (1-\rho^2)^{1/2} \, \Normal(0, \sigma^2_{j})\ , 
\quad \textrm{for} \quad j \geq d+1 \  ,
\end{aligned}
\end{equation}
where we have set $\widehat{\sigma}^2_j := (1/N) \, \sum_{i=1}^{N}\{x_{n-1,j}^i\}^2$. This corresponds to the adaptive SMC approach described in Section \ref{sec:algo} with a $d$-dimensional adaptive  statistics $\xi_n(x)=(x_1^2, \ldots, x_d^2)$. Thus, the $d$ first coordinates of the proposal are adapted to the estimated marginal variance while the ideal variance is used for the remaining coordinates. 
We want to investigate the effect of the amount of adaptation on the 
accuracy of the estimator $\eta_{n}^{N}(\phi)$ for a bounded function $\phi$ that only depends on the $(d+1)$-th coordinate,
\begin{equation*}
%\label{eq:phi}
\phi(x)=\phi(x_{d+1})\ . 
\end{equation*}
Notice that in this simple scenario the Metropolis-Hastings proposal corresponding to the ideal kernel $M_{n,\eta_{n-1}(\xi_n)}$ preserves $\eta_n$ and is thus always accepted; under the ideal kernel, the particles at time $n$ would still be a set of $N$ i.i.d.\@ samples from $\eta_n$. Consequently, any deviation from the $\mathcal{O}(N^{-1/2})$ rate of convergence for the estimator $\eta_{n}^{N}(\phi)$ will be solely due to the effect of the adaptation.

%\emph{Statement of Our Result:}
%
We now investigate in this context the behavior of the difference $\eta_{n}^{N}(\phi)-\eta_{n}(\phi)$. Following the proof of Theorem \ref{theo:wlln} we use the decomposition
\begin{equation*}
[\eta_{n}^{N}-\eta_{n}](\phi) = A(N) + B_1(N) + B_2(N) 
\end{equation*}
where, using the notations of Section \ref{sec:algo}, we have set
$A(N) =  [\eta_n^N -\Phi_{n,N}(\eta_{n-1}^N)](\phi)$, 
$B_1(N)  =\eta_{n-1}^N[Q_{n,N} - Q_n](\phi)$
and
$B_2(N)  = [\eta_{n-1}^N-\eta_{n-1}](Q_n \phi)$.
%B_3(N) & =  -\tfrac{\eta_{n-1}^N(Q_n(\phi_n))}{\eta_{n-1}(G_{n-1})\eta_{n-1}^N(G_{n-1})}\,
%[\eta_{n-1}^N-\eta_{n-1}](G_{n-1})\  .
%
%
Denoting by $\|\!\cdot\! \|_2$ the $L_2$-norm of random variables and 
conditioning upon $\mathcal{F}_{n-1}^{N}$, we have that
\begin{equation}
\label{eq:A}
\| A(N) \|_2^2  = \tfrac{1}{N}\,\Exp\,\big[\,\Var\,[\,\phi(x_n^{1})\,|\,\mathcal{F}_{n-1}^{N}\,]\,\big]
= \mathcal{O}(\tfrac{1}{N})\ .
\end{equation}
For  $B_2(N)$ one can notice that $Q_n(\phi)$ is a bounded mapping from $\bbR^{\infty}$ to $\bbR$, thus
\begin{equation}
\label{eq:B}
\| B_2(N) \|_2^2 = \tfrac{1}{N}\,\mathrm{Var}_{\eta_{n-1}}\,[\,Q_n(\phi)\,] = \mathcal{O}(\tfrac{1}{N})\ . 
\end{equation}
The critical term with regards to the effect of the dimension $d$ on the magnitude of the difference $[\eta_{n}^{N}-\eta_{n}](\phi)$ is $B_1(N)$. An approach similar to Equation \eqref{eq.Q.diff.eta} in the proof of Theorem \ref{thm.clt.unnormalised} yields
\begin{align*} 
B_1(N) 
&= \eta^N_{n-1} [Q_{n,N}-Q_n](\phi) 
=  \eta^N_{n-1}\big( \, \big[ M_{n,N} - 
M_{n} \big]( \phi ) \, \big) \\
&=
\eta^N_{n-1} \big[ \partial_{\xi} M_{n} \phi \big]
\cdot
[\eta^N_{n-1}-\eta_{n-1}](\xi_n) + R 
=:
\widetilde{B}_1(N) + R\ ,  \label{eq:R} 
\end{align*}
for a residual random variable $R$. Controlling the residual term in the above expansion poses 
enormous technical challenges and we restrict our analysis to the main order term $\widetilde{B}_1(N)$.
%We prove the following result in the Appendix. 
%
\begin{prop}
\label{pr:B}
The term $\widetilde{B}_1(N)$ satisfies
\begin{equation*}
\|\widetilde{B}_1(N) \|_2 = \mathcal{O}\big(\tfrac{\sqrt{d}}{N} \big) + \mathcal{O}\big(\tfrac{d}{N^{3/2}}\big)\ . 
\end{equation*}
\end{prop}
\begin{proof}
See the Appendix.
\end{proof}
Proposition \ref{pr:B} combined with \eqref{eq:A}-\eqref{eq:B} suggests that, in a high dimensional setting with $d \gg 1$, it is reasonable to choose $N$ of order $\mathcal{O}(d)$, yielding a mean squared error of order $\mathcal{O}(1/d)$.
Even if this choice of $N$ should be thought of as a minimum requirement for the complete sequential method, 
it could maybe explain the fairly accurate SMC estimates of the marginal expectation obtained in the Navier-Stokes example when $N=500$ and $d\approx 500$; we refer the reader to \cite{kantas} for further simulation studies.

\section{Summary}\label{sec:summ}
This article studies the asymptotic properties of a class of adaptive SMC algorithms; weak law of large numbers and a central limit theorems are established in several settings.
There are several extensions to the work in this article.
First, one could relax the boundedness assumptions used in the paper; our proof technique, also used in \cite{chopin1}, is particularly amenable to this.
Second, an approach to deal with the random stopping of some adaptive SMC algorithms (see Section \ref{sec:annealed})
also needs to be developed. Lastly, one can extend the analysis to the context of adaptive resampling.
%Lastly, the time-stability of the asymptotic variances could be considered.

\subsubsection*{Acknowledgements}

AB and AT were supported by a Singapore MOE grant.
AJ was supported by Singapore MOE grant R-155-000-119-133 and is also affiliated with the risk management institute at the National University of Singapore.

%Fourthly, for applications such as in Section \ref{sec:exam1}, one would like to upper-bound the asymptotic variance w.r.t.~$n$ it would be of interest to derive the conditions for which this holds.

%
% APPENDIX
%
\appendix

\section{Proof of Proposition \ref{pr:B}}

First of all, notice that without loss of generality we can assume that $\sigma_j^2 = const.$. 
We have that:
\begin{align}
\widetilde{B}_1(N) &= \frac{\sqrt{d}}{N}\times \sum_{j=1}^{d}
\Big\{   \tfrac{\sum_{i=1}^{N}\partial_{\xi_j}M_{n,\xi}(\phi)(x_{n-1}^{i})|_{\xi=\eta_{n-1}(\xi_n)}}{\sqrt{N}}\cdot \sqrt{N}\,(\eta_{n-1}^N -
 \eta_{n-1})(\xi_{n,j})   \Big\}/\sqrt{d} \nonumber \\
&\equiv \frac{\sqrt{d}}{N}\times \sum_{j=1}^{d}\big[\, \sqrt{N}\,\eta_{n-1}^{N}
(\bar{\Xi}_{n,j})\cdot \sqrt{N}\eta_{n-1}^N(\bar{\xi}_{n,j})  \,\big]/\sqrt{d} \label{eq:aa}
\end{align}
where we have set $\bar{\Xi}_{n,j}(x)=\partial_{\xi_j}M_{n,\xi}(\phi)(x)|_{\xi=\eta_{n-1}(\xi_n)}$ and 
$\bar{\xi}_{n,j}(x)=
{\xi}_{n,j}(x)-\eta_{n-1}({\xi}_{n,j})$. Clearly, the expectation of the latter variable over $\eta_{n-1}$ is 
zero, but the same is also true for the former one. % as observed in Section \ref{sec:equal}. 
Initially, we will focus on the term $\bar{\Xi}_{n,j}(x)$ as it has some structure which will be exploited in subsequent calculations. Indeed, considering $M_{n,\xi_{j}}(\phi)(x)$, for an arbitrary $\xi_j$ and the rest  $\xi_{k}$, $k\neq j$, at their limiting `correct' values, we have that:
\begin{equation}
\label{eq:exp}
M_{n,\xi_{j}}(\phi)(x) = \Exp\,[\,\phi(x_{d+1}')\,|\,x\,]  = \phi(x_{d+1}) + \Exp\,[\,a(x_{j},\xi_j,Z_{j})\,|\,x_{j}\,]\,\Delta \phi(x_{d+1})
\end{equation}
where we have set $\Delta \phi(x_{d+1}) = \Exp\,[\,\phi(x'_{d+1})-\phi(x_{d+1})\,|\,x_{d+1}\,]$; 
$x_{d+1}'$ denotes the Metropolis-Hastings proposal for the $(d+1)$-th co-ordinate as specified in \eqref{eq:M} ; $a(x_j,\xi_j,Z_j)$ denotes the Metropolis-Hastings acceptance probability  which depends only on the current position $x_j$,
the (arbitrary) scaling choice $\xi_j$ and the noise $Z_j\sim \mathcal{N}(0,1)$ for simulating the proposal 
for the $j$-th co-ordinate assuming a scaling $\xi_j$ (that is, we have $x_{j}'=\rho x_j + \sqrt{1-\rho^2}\,\xi_j^{1/2}\,Z_j$). 
%The  quantity quantity (given
%that proposals preserve their marginal target apart from the one for the $j$-th co-ordinate for which the statistic $\xi_j$ is 
%not the ideal limiting one) depends only on $x_j, \xi_j$ and the noise in the proposal for the $j$-th co-ordinate $Z_j$ - 
We will give the explicit formula for $a(\cdot)$ below.
Notice that due to the proposal for $x_{d+1}$ preserving the target marginally at the $(d+1)$-th co-ordinate, we 
have that $\Exp_{\eta_{n-1}}\,[\,\Delta \phi(x_{d+1})\,] = 0$.
Recall that $\bar{\Xi}_{n,j}(x) = \partial_{\xi_j}M_{n,\xi_j}(\phi)(x)|_{\xi_j=\eta_{n-1}(\xi_{n,j})}$, thus to check for the differentiability of the  mapping
 $\xi_j \mapsto  \Exp\,[\,a(x_{j},\xi_j,Z_{j})\,|\,x_{j}\,]$ we can only resort to analytical calculations, starting from the fact that
(after some algebraic manipulations):
\begin{equation*}
a(x_j,\xi_j,Z_j) = 1 \wedge \exp\Big\{   -\tfrac{1}{2}\big( \,\xi_j^{-1} -\sigma_j^{-2}\big)\big(x_j^2 - \big\{\rho\, x_j + \sqrt{1-\rho^2}\,\xi_j^{1/2}Z_j\big\}^2\,\big)  \Big\}\ . 
\end{equation*}
%f
After a lot of cumbersome analytical calculations (which are omitted for brevity) we can integrate out $Z_j$ and find that i) the derivative  $D(x_j, \eta_{n-1}(\xi_{n,j})) = \partial_{\xi_j}\Exp\,[\,a(x_{j},\xi_j,Z_{j})\,|\,x_{j}\,]|_{\xi_j=\eta_{n-1}(\xi_{n,j})}$ exists; ii) $D(x_j, \eta_{n-1}(\xi_{n,j}))$, with $x_j\sim \mathcal{N}(0,\sigma_j^2)$, has a finite second moment.
Thus, continuing from \eqref{eq:exp} we have:
\begin{equation}
\label{eq:fact}
\bar{\Xi}_{n,j}(x) = \partial_{\xi_j}M_{n,\xi_j}(\phi)(x)|_{\xi_j=\eta_{n-1}(\xi_{n,j})} =  D(x_j,\eta_{n-1}(\xi_{n,j}))\,\Delta\phi(x_{d+1})\  .
\end{equation}
The factorisation in (\ref{eq:fact}) will be exploited in the remaining calculations.

Continuing from (\ref{eq:aa}),  we now have that:
\begin{align}
\|\tfrac{N}{\sqrt{d}}&\widetilde{B}_1(N) \|_2^2 = 
\tfrac{1}{d}\sum_{j=1}^{d}N^2\,\Exp\,\big[\, \{\eta_{n-1}^N(\bar{\Xi}_{n,j})\}^2\, \{\eta^N(\bar{\xi}_{n,j})\}^2\big]
\nonumber
\\
&+ \tfrac{1}{d} \sum_{\substack{j,k=1,2,\ldots, d\\ j\neq k}} N^2\,\Exp\,\big[\, \eta_{n-1}^N(\bar{\Xi}_{n,j})\, \eta_{n-1}^N(\bar{\xi}_{n,j})\,\eta_{n-1}^N(\bar{\Xi}_{n,k})\,\eta_{n-1}^N(\bar{\xi}_{n,k})  \,\big] 
\nonumber
\\ 
 &\qquad \qquad \qquad =: T_1 + T_2 \ . \label{eq:B1}
\end{align}
The following zero-expectations obtained for terms involved in $T_1$, $T_2$ are a direct consequence 
of the fact that $\bar{\xi}_{n,j}(x)$ only depends on $x_j$ and has zero expectation under $\eta_{n-1}$, and that 
$\bar{\Xi}_{n,j}(x)$ only depends on $x_j$, $x_{d+1}$ through the product form in (\ref{eq:fact}) with the $x_{d+1}$-term 
having zero-expectation; critically, recall that particles $x_{n-1,j}^{i}$ are independent over both $i, j$.
Focusing on the $T_1$-term and the expectation $\Exp\,\big[\, \{\eta_{n-1}^N(\bar{\Xi}_{n,j})\}^2\, \{\eta_{n-1}^N(\bar{\xi}_{n,j})\}^2\big]$ we note that all 4-way product terms arising after replacing $\eta_{n-1}^{N}$ with its sum-expression will have expectation 0, 
except for the ones that involve cross-products of the form 
$\{\bar{\Xi}_{n,j}(x_{n-1}^i)\}^2 \times \{\bar{\xi}_{n,j}(x_{n-1}^{i'})\}^2$, thus:
\begin{equation}
\label{eq:T1}
T_1 =  \tfrac{1}{d}\sum_{j=1}^{d} N^2\cdot \tfrac{1}{N^4}\cdot \mathcal{O}(N^2) = \mathcal{O}(1)\ .
\end{equation}
Then, moving on to the $T_2$-term, notice that  all 4-way products in the expectation term $\Exp\,\big[\, \eta_{n-1}^N(\bar{\Xi}_{n,j})\, \eta_{n-1}^N(\bar{\xi}_{n,j})\,\eta_{n-1}^N(\bar{\Xi}_{n,k})\,\eta_{n-1}^N(\bar{\xi}_{n,k})  \,\big]$ have expectation $0$, except for the products involving the same particles
$\bar{\Xi}_{n,j}(x_{n-1}^{i})\, \bar{\xi}_{n,j}(x_{n-1}^{i})\,\bar{\Xi}_{n,k}(x_{n-1}^{i})\,\bar{\xi}_{n,k}(x_{n-1}^{i})$.
Thus, we have that:
\begin{equation*}
T_2 = \tfrac{1}{d}\,\sum_{j,k=1,j\neq k}^{d}N^2\cdot  \tfrac{1}{N^4}\cdot \mathcal{O}(N) = \mathcal{O}(\tfrac{d}{N})
\end{equation*}
Thus, overall we have that:
\begin{equation}
\label{eq:T2}
\|\widetilde{B}_1(N) \|_2 = \mathcal{O}(\tfrac{\sqrt{d}}{N}) + \mathcal{O}(\tfrac{d}{N^{3/2}})\ . 
\end{equation}
Results (\ref{eq:T1}), (\ref{eq:T2}), used within (\ref{eq:B1}) complete the proof.

\end{document}